\documentclass{article}

\PassOptionsToPackage{numbers, compress, square}{natbib}
\usepackage[preprint]{neurips_2026}
\usepackage[utf8]{inputenc} 
\usepackage[T1]{fontenc}    
\usepackage{hyperref}       
\usepackage{url}            
\usepackage{booktabs}       
\usepackage{amsfonts}       
\usepackage{nicefrac}       
\usepackage{microtype}      
\usepackage{xcolor}         

\usepackage{amsmath}
\usepackage{mathrsfs}
\usepackage{amsthm}
\usepackage{amssymb}

\usepackage{algorithm2e}
\RestyleAlgo{ruled}

\usepackage[capitalize]{cleveref}
\usepackage{mathtools}

\makeatletter
\AddToHook{cmd/appendix/before}{\def\cref@section@alias{appendix}\def\cref@subsection@alias{appendix}}
\makeatother

\DeclareMathOperator{\KLop}{KL}

\newcommand{\KL}[2]{%
  \mathchoice
    {\KLop\!\left(#1\,\middle\|\,#2\right)}
    {\KLop(#1\,\Vert\,#2)}
    {\KLop(#1\,\Vert\,#2)}
    {\KLop(#1\,\Vert\,#2)}
}
\newcommand{\Rdiv}[3][\eps]{%
  \mathfrak{R}_{#1}\!\left(#2\,\middle\|\,#3\right)%
}
\DeclareMathOperator{\TVop}{TV}
\newcommand{\TV}[2]{
  \mathchoice
    {\TVop\!\left(#1,#2\right)}
    {\TVop(#1,#2)}
    {\TVop(#1,#2)}
    {\TVop(#1,#2)}
}

\newcommand{\Null}{\mathbb{P}}

\newcommand{\E}{\mathbb{E}}
\newcommand{\Alt}{\mathbb{Q}}
\newcommand{\M}{\mathcal{M}}

\newcommand{\opt}{\Rdiv{\Alt}{\Null}}

\newcommand{\tsllr}{\text{tsLR}}
\newcommand{\eps}{\varepsilon}
\iftrue  
  \newcommand{\ben}[1]{\textcolor{purple}{#1}}
  \newcommand{\tg}[1]{[\textcolor{blue}{TG: #1}] }
  \newcommand{\ar}[1]{[\textcolor{red}{AR: #1}] }
  \newcommand{\gb}[1]{[\textcolor{brown}{GB: #1}] }
  
\else
  \newcommand{\ben}[1]{}
  \newcommand{\tg}[1]{}
  \newcommand{\gb}[1]{}
  \newcommand{\ar}[1]{}
\fi

\newcommand{\cI}{\mathcal{I}}
\newcommand{\cM}{\mathcal{M}}
\newcommand{\cO}{\mathcal{O}}
\newcommand{\cX}{\mathcal{X}}

\DeclareMathOperator*{\argmin}{\arg\min}
\DeclareMathOperator{\sign}{sign}

\newtheorem{theorem}{Theorem}[section]
\newtheorem{corollary}{Corollary}[theorem]
\newtheorem{lemma}{Lemma}[theorem]
\newtheorem{definition}{Definition}[section]
\newtheorem{remark}{Remark}[section]
\newtheorem{proposition}{Proposition}[section]

\title{Optimal Rates for Differentially Private Hypothesis Testing with E-values}

\author{%
  Ben Jacobsen \\
  University of Wisconsin-Madison\\
  \texttt{bjacobsen3@wisc.edu} \\
  \And 
  Tomas Gonzalez \\
  Carnegie Mellon University\\
  \texttt{tcgonzal@cs.cmu.edu} \\
  \And 
  Gavin Brown \\
  University of Wisconsin--Madison \\
  \texttt{gavin.brown@wisc.edu} \\
  \And 
  Kassem Fawaz \\
  University of Wisconsin--Madison \\
  \texttt{kfawaz@wisc.edu} \\
  \And 
  Aaditya Ramdas \\
  Carnegie Mellon University \\
  \texttt{aramdas@cs.cmu.edu}
}

\begin{document}

\maketitle

\begin{abstract}
  E-values have attracted considerable interest in recent years as flexible tools for enabling anytime-valid and adaptive data analysis.
  Hypothesis testing is at the core of many of these applications, which can often involve private or sensitive data.
  In this work, we answer a simple but important question: given two distributions $\mathbb{P}$ and $\mathbb{Q}$, 
  what is the maximum achievable e-power when testing $X\sim \mathbb{P}^n$ against $X\sim\mathbb{Q}^n$ with e-values that satisfy $\varepsilon$-differential privacy? 
  We characterize the optimal rate for this problem and provide an algorithm which matches it exactly.
  In the sequential setting, when observations arrive one-by-one and the analyst chooses when to halt, we give matching upper and lower bounds on the stopping times of any private e-process.
  Numerical experiments confirm the practicality of our algorithms, which require less data than the recently proposed DP-SPRT across a range of sequential testing problems and privacy levels.
\end{abstract}

\section{Introduction}

Hypothesis testing is ubiquitous in modern scientific practice and also one of the most fundamental problems in statistical inference. 
We consider simple hypothesis testing: an algorithm receives a dataset $X\in \cX^n$ and must evaluate whether it was sampled from a distribution $\Null$ (the null hypothesis) or $\Alt$ (the alternative).
In the classical, decision-theoretic setting, we wish to output a single decision $\phi(X) \in \lbrace 0,1 \rbrace$, where the output $\phi=1$ is interpreted as rejecting the null hypothesis in favor of the alternate. 
To prevent false discoveries, we ask that $\Null[\phi(X) = 1] \leq \alpha$ for some fixed significance level $\alpha$. 
The performance of the test is then measured by its true-positive rate or \emph{power}, $\Alt[\phi(X) = 1]$.

Recently, a complementary paradigm has emerged based on the concept of \emph{e-values}~\citep{ramdasHypothesisTestingEvalues}. 
Instead of binary decisions, algorithms for hypothesis testing in this setting output a real-valued, non-negative value $E$ which
quantifies the degree to which $X$ provides evidence for the alternate hypothesis over the null. 
We call $E$ an \emph{e-variable} and refer to realizations of $E$ as \emph{e-values}, but gloss over this distinction when speaking informally.
Larger e-values present more evidence for the alternative, because we expect the e-value to be small under the null, as per its definition.
\begin{definition}
    An \emph{e-variable for $\Null$} is a non-negative random variable $E$ satisfying $\E^\Null[E] \le 1$. 
\end{definition}
One can directly convert $E$ into a hypothesis test, rejecting the null if $E$ exceeds $1/\alpha$; this is a level-$\alpha$ test due to Markov's inequality. 
We measure the e-value's performance by its (e-)power.
\begin{definition}
    Let $E$ be an e-variable for $\Null$. Its \emph{e-power against $\Alt$} is defined as $\E^\Alt[\log E]$. We say that E is \emph{log-optimal} if $\E^\Alt[\log(E'/E)] \leq 0$ for any other e-variable $E'$ for $\Null$.
\end{definition}

A major advantage of this approach is that 
statistical inferences drawn from (products of) e-values often remain valid even when decisions about data collection and experimental design are made adaptively, whereas those same conditions typically invalidate type-I error control in classical methods~\citep{ramdasHypothesisTestingEvalues,grunwaldSafeTesting2020}. 
There is thus good reason to hope that more widespread use of e-values and their sequential analog, e-processes, could mitigate issues related to the misuse of p-values which have contributed to the ongoing scientific reproducibility crisis~\citep{ter_Schure_2019,amrhein2019scientists,shafer2021testing,grunwaldSafeTesting2020,ramdas2023game,chugg2026evaluesstatisticalevidencecomparison,huang2026controlling}. Methods based on e-values are also broadly useful for modern statistical workflows in which the number and timing of analyses are rarely fixed in advance such as adaptive clinical trials \cite{adaptive-clinical-trials}, online active experimentation \cite{waudby2024anytime-valid-contex-bandits, martinez-taboada2026vectorvalued}, online multiple testing \cite{online_multiple_testing_evalues}, sequential change detection \cite{edetectors} and auditing \cite{risk_auditing, fairness_auditing, privacy-auditing}. 

Many statistical applications are driven by sensitive individual-level data, however. This motivates the need for methods which additionally satisfy \emph{differential privacy} (DP), a rigorous and widely used definition of data privacy which guarantees that the result of a statistical analysis cannot depend too much on any one person's data \cite{dwork2014algorithmic,fioretto2025differential,abowd2018us}. 
\begin{definition}
    For $\eps\ge 0$, an algorithm $\M:\cX^n\to\cO$ is \emph{$\eps$-differentially private} if for all datasets $X$ and $X'$ differing in one entry and all $S\subseteq \cO$ we have 
    $\Pr[\cM(X)\in S]\le e^\eps \Pr[\M(X')\in S]$.
\end{definition}

There is a large body of literature investigating differentially private hypothesis testing in the classical decision-theoretic setting~\citep{canonne2019structure,5360513,kazan2023testtestsframeworkdifferentially,kifer2016newclassprivatechisquare,wang2017revisitingdifferentiallyprivatehypothesis,swanberg2019improveddifferentiallyprivateanalysis,Covington_2024}, including several works that study sequential testing problems~\citep{cummings2018differentiallyprivatechangepointdetection,zhang2022DPSPRT,michel2026dpsprtdifferentiallyprivatesequential,nikolakakis2022quantilemultiarmedbanditsoptimal,azize2026differentiallyprivatebestarmidentification}. 
We are aware of only three papers on differential privacy and e-values. 
\citet{waudbysmith2024nonparametricextensionsrandomizedresponse} and \citet{saha2026optimalevariablesconstraints} consider \emph{local} DP, wherein privacy protections are applied to individual data points before aggregation~\citep{kasiviswanathan2011can}. 
But while the privacy guarantees provided by local DP are very strong, they can also come at the cost of substantial utility degradation and are not appropriate for all statistical applications.

Our work therefore focuses on the more commonly used \emph{central} model of DP, where the privacy protections apply to the output of a  trusted aggregator. 
To our knowledge, the only work that considers e-values in this setting is that of \citet{csillagDifferentiallyPrivateEValues2025a}.
The algorithms described in that work only apply to testing distributions $\Null$ and $\Alt$ with bounded likelihood ratios, however, which fails to hold even in the simple case of testing Gaussians.

We aim to provide a more complete understanding of DP e-values for hypothesis testing. 
Within this context, our major contributions are:
\begin{itemize}
    \item An instance-specific hardness result describing the maximum possible e-power achievable by any $\varepsilon$-DP e-value when testing $\Null$ against $\Alt$ (\cref{thm:e-power-ub}).
    \item The construction of an $\eps$-DP e-value which achieves optimal e-power for any arbitrary pair of distributions, up to lower order terms (\cref{thm:basic-batch}).
    \item An instance-specific lower-bound on the expected stopping time of any sequential test based on thresholding an $\eps$-DP e-process, i.e., the amount of data which must be collected on average to achieve a given level of statistical power (\cref{thm:agr-ub}).
    \item An efficient algorithm for computing $\eps$-DP e-processes which simultaneously matches this bound, up to a constant multiplicative factor, for \emph{all} stopping times larger than some tunable minimal threshold (\cref{alg:batch2proc}, \cref{thm:sequential-lb}). In \cref{sec:simulation}, we compare \cref{alg:batch2proc} against the recently proposed DP-SPRT of~\citet{michel2026dpsprtdifferentiallyprivatesequential} and find that our algorithm achieves consistently earlier stopping times across a range of settings.
\end{itemize}

Given the central importance of simple hypothesis tests, our results also provide a foundation for investigating many other interesting problems involving robust, adaptive, and privacy-preserving statistical methods based on e-values. For instance, a natural next step would be to investigate $\eps$-DP e-values for composite hypothesis testing; in that setting, it is often possible to directly apply hardness results for simple hypothesis testing, such as our \cref{thm:e-power-ub}, via so-called least-favorable pairs~\citep{huber1973minimax}. A similar idea arises in statistical estimation, where the difficulty of estimating a parameter $\theta$ to error $\pm \alpha$ is strongly connected to worst-case difficulty of testing a distribution $D_\theta$ against $D_{\theta'}$ for $|\theta - \theta'| \leq \alpha$~\citep{donoho1991geometrizingII}. It is therefore likely that the ideas we present could be extended to better understand the complexity of e.g.\ anytime-valid confidence sequences~\citep{ramdasAdmissibleAnytimevalidSequential2022,ramdas2023game} under central DP, much as the foundational work of~\citet{canonne2019structure} on private simple hypothesis testing in the decision-theoretic setting helped contribute to breakthroughs in instance-optimal DP estimation algorithms~\citep{mcmillan2022instanceoptimaldifferentiallyprivateestimation,universally-instance-optimal-mechanisms}.

\subsection*{Preliminaries and Notation}

The formal definition of differentially private e-variables requires some care: DP is a property of a randomized algorithm $\M$ while the e-variable property is about the output $E=\M(X)$, a random variable.
We draw a clear distinction between these two objects.
\begin{definition}[Differentially private e-variable and e-power]
    We say that a non-negative random variable $E = \M(X)$ is an \emph{$\eps$-DP e-variable for $\Null$} whenever $\M$ is $\varepsilon$-DP and $\E^{\Null^n}[\M(X)] \leq 1$, where the expectation is taken with respect to both the input data $X \sim \Null^n$ and the internal randomness of $\M$. The e-power of $E$ against an alternative $\Alt$ is $\E^{\Alt^n}[\log \M(X)]$.
\end{definition}

 We use capital letters like $P$ and $Q$ for distributions and corresponding lowercase letters like $p$ and $q$ for their densities. $\M(X)$ is a random variable representing the output of $\M$ on input $X$, while $\M(\Null^n)$ represents the output distribution of $\M$ when $X \sim \Null^n$. We use the notation $\E^{R^n}[f]$ to represent the expectation of $f$ when the underlying dataset $X$ is an i.i.d.\ sample of $n$ points from $R$. When $n$ is clear from context, we will sometimes abbreviate this to $\E^R[f]$. If $f$ is randomized, then the expectation is taken with respect to both the input data and the internal randomness of $f$ unless otherwise stated. 
 
We mainly state our results in terms of continuous distributions and ignore measure-theoretic details in the main body of the paper. We collect a number of standard definitions and lemmas related to differential privacy, information theory, and e-values in \cref{app:lemmas-and-defs}. Formal proofs of most of our results are deferred to \cref{sec:deferred_proofs}.

\section{Batch Setting}\label{sec:batch}
\subsection{Characterization of the Optimal E-power under Pure DP}
\label{sec:batch_optimal}

In this section, we present our results for the batch setting, i.e., when the analyst receives all the data at once.
Without privacy, the maximum achievable e-power is straightforward: if $E$ is an e-variable for $\Null$, then $\frac{1}{n}\E^{\Alt^n}[\log E] \leq \KL{\Alt}{\Null}$ and this bound is achieved exactly by the likelihood ratio e-variable $R = \frac{d\Alt^n}{d\Null^n}(x)$~\citep{ramdasHypothesisTestingEvalues}. We are interested in understanding the corresponding optimal rate under central DP constraints:
\begin{equation}
    \mathfrak{R}_{n,\eps}(\Alt~\Vert~\Null) \coloneq \sup_{\varepsilon\text{-DP e-variables } E \text{ for }\Null} \frac{\E^{\Alt^n}[\log E]}{n},\qquad
    \Rdiv{\Alt}{\Null} \coloneq \lim_{n\to\infty} \mathfrak{R}_{n,\eps}(\Alt~\Vert~\Null).
    \label{eq:rate_def}
\end{equation}

In our search for optimal private e-values, we find it useful to decouple the privacy properties from the e-variable properties. 
More formally, consider fixing an arbitrary $\eps$-DP mechanism $\cM$.
How should we turn its output into an e-variable?
We observe that we now have a different testing problem: distinguishing null $\cM(\Null^n)$ from alternate $\M(\Alt^n)$.
Crucially, by the post-processing property of differential privacy, we can consider how to construct an e-variable from the output of $\M$ without considering privacy explicitly.

Absent privacy constraints, the optimal e-value is simply the likelihood ratio. 
Its e-power is simply the Kullback--Leibler divergence $\KL{\Alt}{\Null}$.
Thus, for any fixed DP mechanism $\M$, we can post-process its output into an e-variable $E_\M$ with e-power
\[
    \mathbb{E}^{\Alt^n}[\log E_\M] 
= \mathop{\mathbb{E}}\limits_{\substack{
  X \sim \Alt^n \\
  Y \sim \cM(X)
}}\left[\log \frac{\mathrm{d}\M(\Alt^n)}{\mathrm{d}\M(\Null^n)}(Y) \right] 
= \KL{\M(\Alt^n)}{\M(\Null^n)},
\]
and this is the best way to process $\M$'s output.
Immediately, we have the following characterization of the rate.

\begin{proposition}\label{prop:characterization}
For any null $\Null$, alternate $\Alt$, and $\varepsilon > 0$, the following equality holds:
\[\Rdiv{\Alt}{\Null} = \lim_{n \to \infty} \sup_{\M \hspace{.1cm}\eps\text{-DP }}  \frac{\KL{\M(\Alt^n)}{\M(\Null^n)}}{n}.\]
\end{proposition}

A maximization problem similar to the one defined by \cref{prop:characterization} has previously been studied in the context of local DP~\citep{duchi2014localprivacydataprocessing}, where it is known that the optimum is achieved by a specific family of two-value ``staircase'' distributions~\citep{kairouz2015extremalmechanismslocaldifferential}. 
The problem is much less well understood under central DP, however, and so it is not immediately evident how to use this result to argue optimality.

Our core results for the batch setting attack this problem from two directions: an upper bound in \cref{sec:batch_rate_upper} and lower bound in \cref{sec:batch_lower_bound}.
The attacks are coordinated: both depend on the analysis of a particular intermediate distribution $\widetilde Q$ that serves as a bridge between $\Null$ and $\Alt$. We present the construction of $\widetilde Q$ in \cref{sec:batch_construction}. 

Our bounds match exactly, up to lower order terms, meaning that the private e-values we construct are asymptotically log-optimal for all pairs of distributions. As a consequence, we derive the following strong duality result characterizing the maximum KL divergence between central DP mechanisms, which may be of independent interest:

\begin{theorem}\label{thm:duality}
    Let $\mathcal{D}$ denote the set of probability distributions over $\mathcal{X}$. For any $\Null,\Alt \in \mathcal{D}$ and $\eps > 0$, the following equality holds:
    \begin{equation}\label{eq:duality}
    \lim_{n\to\infty}\sup_{\M \hspace{.1cm}\eps\text{-DP }}   \frac{\KL{\M(\Alt^n)}{\M(\Null^n)}}{n} = \inf_{Q' \in \mathcal{D}} \KL{Q'}{\Null} + \eps \TV{Q'}{\Alt}.
\end{equation}
\end{theorem}

We remark that right-hand side of \eqref{eq:duality} can be interpreted as a relaxation of the problem of finding the worst corrupted alternate in robust statistics~\citep{huber1973capacities}, where it is known that the optimal test statistic is obtained by clipping the likelihood ratio $d\Alt/d\Null$ to some interval $[c',c'']$~\citep{huber1965robust}. Our problem differs in that we allow $\TV{Q'}{\Alt}$ to vary with $\Null$, $\Alt$, and $\eps$ instead of fixing a particular level of contamination in advance. Nonetheless, we show that a (different) clipped version of the likelihood ratio is similarly optimal for testing under central DP.

\subsection{Constructing the Intermediate Distribution}
\label{sec:batch_construction}

In this section we present our construction of $\widetilde Q$ and some associated quantities upon which our later results rely.
We present the construction tersely, mentioning only that $\widetilde Q$ is the distribution where the infimum on the right-hand side of \eqref{eq:duality} is realized.
We give further explanations about $\widetilde Q$ alongside the relevant arguments and visualize the construction of $\widetilde Q$ for a specific pair of distributions in \cref{fig:construction}.

Fix null $\Null$, alternate $\Alt$, and let $\lambda \in \mathbb{R}$ be arbitrary. We introduce the quantities $c_1(\lambda) = e^{-\eps/2 + \lambda - 1}, c_2(\lambda) = e^{\eps/2 + \lambda - 1}$, and define the following subsets of $\mathcal{X} $: 
\begin{align*}
    &A_\lambda = \lbrace x: q(x) < c_1(\lambda) p(x) \rbrace,
    &B_\lambda = \lbrace x : q(x) > c_2(\lambda) p(x) \rbrace, \qquad
    &M_\lambda = \mathcal{X} \setminus (A_\lambda \cup B_\lambda).
\end{align*}

We then define the key quantity $\lambda^*$ as the solution to the following implicit equation
\begin{equation}\label{eq:lambda-star}
    f(\lambda) \coloneq c_1(\lambda) \Null(A_\lambda) + \Alt(M_\lambda) + c_2(\lambda) \Null(B_\lambda) = 1.
\end{equation}

We establish the existence of a solution in \cref{sec:deferred_batch}, and will therefore 
proceed to use $A$, $B$, and $M$ as shorthand for $A_{\lambda^*}$, $B_{\lambda^*}$, and $M_{\lambda^*}$ respectively. Similarly, we will use $c_1$ and $c_2$ as shorthand for $c_1(\lambda^*)$ and $c_2(\lambda^*)$.

We are now ready to present our new distribution $\widetilde Q$, whose density is defined as follows:
\begin{equation}
    \widetilde q(x) =
    \begin{cases}
        c_1\cdot  p(x) & x \in A, \\
        q(x) & x \in M, \\
        c_2 \cdot  p(x) & x \in B. \\
    \end{cases}
\end{equation}

This integrates to 1 by the definition of $\lambda^*$, ensuring that $\widetilde Q$ is in fact a probability distribution. 

\subsection{Bounding the Optimal Rate}
\label{sec:batch_rate_upper}

Our proof technique for deriving an upper bound on e-power is based on techniques of \citet{balle2018privacyamplificationsubsamplingtight}, which were subsequently extended by \citet{canonne2019structure} to characterize optimal private algorithms for symmetric hypothesis tests.  
A key component of our approach is the following coupling and group-privacy lemma, which applies to general intermediate distributions.
\begin{lemma}\label{lem:decomp}
    Let $\Null, \Alt, Q' \in \mathcal{D}$, and let $\M: \mathcal{X}^n \to \mathcal{O}$ be an arbitrary $\eps$-DP mechanism. Then:
    \begin{equation*}
        \frac{1}{n}\KL{\M(\Alt^n)}{\M(\Null^n)} \leq \KL{Q'}{\Null} + \varepsilon \TV{Q'}{\Alt}.
    \end{equation*}
\end{lemma}

\begin{proof}
  Our proof is based on the construction of couplings between several key random variables: the original dataset $X$, the output of our algorithm $\M(X)$, and a `shadow dataset' $\tilde X$ which we will introduce shortly. We will define one coupling of these random variables under the alternate (i.e.\ when $X \sim \Alt^n$) and a different one under the null (i.e.\ when $X \sim \Null^n)$.
  
  Under the alternate, we use the total-variation coupling $\gamma$ between $\Alt$ and $Q'$ (\cref{lem:tv-coupling}). That is, for $(X_i, \tilde X_i) \sim \gamma$, we have $\gamma(X_i \neq \tilde X_i) = w \coloneq \TV{Q'}{\Alt}$. We sample the true dataset $X$ and the shadow dataset $\tilde X$ jointly from $\gamma^n$, so that $X \sim \Alt^n$ and $\tilde X \sim Q'^n$. We additionally define a sequence of indicator random variables $B = B_{1:n}$, where $B_i \coloneq \mathbb{I}[X_i \neq \tilde X_i] \sim \text{Bern}(w)$.
  
  Under the null, we instead sample $X\sim \Null^n$ and set $\tilde X = X$. For technical reasons that will become clear shortly, we will not set the distribution of $B_i$ under the null to be an indicator variable. Instead, we require that $\Null(B_i=1 \mid \tilde X=\tilde x) = \Alt(B_i = 1 \mid \tilde X=\tilde x)$, i.e., the conditional distributions match.

  For a random variable $W$,  we use $\Alt_W$ to denote the distribution of $W$ under the alternate, and likewise for $\Null_W$ under the null. The conditional distribution of $W$ given a realized value $z$ of a random variable $Z$ is written as $\Alt_{W\mid z}$ (resp. $\Null_{W \mid z}$).
  
  With this setup done, our goal is to bound $\KL{\Alt_{\M(X)}}{\Null_{\M(X)}}$. We begin by applying the data processing inequality:
\begin{align*}
    \KL{\Alt_{\M(X)}}{\Null_{\M(X)}} 
    &\leq \KL{\Alt_{\M(X),\tilde X, B}}{\Null_{\M(X), \tilde X, B}}.
\end{align*}

By the chain rule for KL divergence (\cref{lem:kl_chain_rule}), it follows that:
\begin{align}
     \KL{\Alt_{\M(X)}}{\Null_{\M(X)}}
     &\leq 
     \KL{\Alt_{\tilde X, B}}{\Null_{\tilde X, B}} \label{eq:DPI-term}\\
     &\quad+ \E_{\tilde x, b \sim \Alt_{\tilde X, B}}\left[ \KL{\Alt_{\M(X)\mid \tilde x, b}}{\Null_{\M(\tilde x) \mid \tilde x, b}} \right] \label{eq:GP-term},
\end{align}
where we used the fact that for $X \sim \Null_{X \mid \tilde x}$, $X = \tilde x$ by construction and thus $\M(X) = \M(\tilde x)$. Next, we will bound each of these terms separately. To bound \cref{eq:DPI-term}, we apply chain rule again to get: 
\begin{equation*}
    \KL{\Alt_{\tilde X, B}}{\Null_{\tilde X, B}}
    = \KL{\Alt_{\tilde X}}{\Null_{\tilde X}} + 
    \E_{\tilde x \sim \Alt_{\tilde X}}\left[ \KL{\Alt_{B \mid \tilde x} }{\Null_{B \mid \tilde x}} \right].
\end{equation*}

Because we defined our coupling so that the conditional distribution $B$ given $\tilde x$ is the same under both hypotheses, $\KL{\Alt_{B \mid \tilde x} }{\Null_{B \mid \tilde x}} = 0$ for any $\tilde x$, and so we are left with $\KL{\Alt_{\tilde X}}{\Null_{\tilde X}}$. But by construction, $\tilde X \sim (Q')^n$ under $\Alt$ and $\tilde X \sim \Null^n$ under $\Null$, and so   $\KL{\Alt_{\tilde X}}{\Null_{\tilde X}} = n \KL{Q'}{\Null}$.

Next, to bound the expectation in \cref{eq:GP-term}, we will use group privacy (\cref{lem:group-privacy}). To start, fix some particular value of $\tilde x, b$. We can then write $\Alt_{\M(X) \mid \tilde x, b}$ as a mixture distribution with components $\Alt_{\M(x) \mid x, \tilde x, b}$ indexed by $x \sim \Alt_{X \mid \tilde x, b}$. 
We will also write $\Null_{\M(\tilde x) \mid \tilde x, b}$ as a mixture with matching weights using a trivial decomposition with all components equal to $\Null_{\M(\tilde x) \mid  \tilde x, b}$.
Then, by the joint convexity of KL divergence (\cref{lem:kl_joint_convexity}), we have that:
\begin{align*}
    \KL{\Alt_{\M(X) \mid \tilde x, b}}{\Null_{\M(\tilde x) \mid \tilde x, b}} 
    &\leq \E_{x \sim \Alt_{X \mid \tilde x, b}} \left[\KL{\Alt_{\M(x) \mid x, \tilde x, b}}{\Null_{\M(\tilde x) \mid \tilde x, b}} \right]  \\
    &= \int_{\mathcal{X}^n} \int_{\mathcal{O}} \log \frac{d\Alt_{\M(x)\mid x,\tilde x,b}}{d\Null_{\M(\tilde x)\mid \tilde x,b}}~d\Alt_{Y \mid x, \tilde x, b}~d\Alt_{X \mid \tilde x,b}.
\end{align*}

Conditioned on $x$ and $\tilde x$, the likelihood ratio in the integrand depends solely on the internal randomness of $\M$ and not on $\Null$ or $\Alt$. Then, because $\M$ is $\eps$-DP and $x, \tilde x$ differ on exactly $\sum_{i=1}^n b_i$ components, we have by group privacy that $\frac{d\Alt_{\M(x)\mid x, \tilde x, b}}{d\Null_{\M(\tilde x) \mid \tilde x, b}} \leq \exp(\sum_{i=1}^n b_i\eps)$ pointwise. Hence,
\begin{equation}\label{eq:jc+gp}
\KL{\Alt_{\M(X) \mid \tilde x, b}}{\Null_{\M(\tilde x) \mid \tilde x, b}}
    \leq \sum_{i=1}^n b_i\eps.
\end{equation}

Finally, we have that $\E_{\tilde x, b \sim \Alt_{\tilde X, B}}\left[\sum_{i=1}^n b_i\eps \right] = nw\eps = n\eps \TV{Q'}{\Alt}$ by linearity of expectation. Combining our bounds, we conclude that
\begin{equation*}
    \frac{1}{n}\KL{\Alt_{\M(X)}}{\Null_{\M(X)}} = \frac{1}{n}\KL{\M(\Alt^n)}{\M(\Null^n)} \leq \KL{Q'}{\Null} + \TV{Q'}{\Alt}\eps,
\end{equation*}
as desired.
\end{proof}

Having proved \cref{lem:decomp}, the following weak duality version of \cref{eq:duality} is an immediate corollary:
\begin{corollary}\label{cor:one-sided}
    Let $\Null, \Alt \in\mathcal{D}$ and $\eps > 0$. Then:
    \begin{equation*}
        \sup_{\M \hspace{.1cm}\eps\text{-DP }} \frac{1}{n} \KL{\M(\Alt^n)}{\M(\Null^n)} \leq \inf_{Q' \in \mathcal{D}} \KL{Q'}{\Null} + \eps \TV{Q'}{\Alt}.
    \end{equation*}
\end{corollary}

Next, we show that $\widetilde Q$ is the unique distribution where the infimum on the right-hand side of \cref{cor:one-sided} is realized. The proof is deferred to \cref{sec:deferred_batch}.

\begin{theorem}\label{thm:e-power-ub}
    For any pair of distributions $\Null$ and $\Alt$ and any $\varepsilon > 0$, we have:
    \begin{equation*}
        \inf_{Q' \in \mathcal{D}} \KL{Q'}{\Null} + \eps \TV{Q'}{\Alt}  = \KL{\widetilde Q}{\Null} + \eps \TV{\widetilde Q}{\Alt}.
    \end{equation*}
\end{theorem}

Taken together, these results give us an explicit, computable upper bound on the e-power of any $\eps$-DP e-variable for $\Null$ against $\Alt$. In the following section, we will construct a private e-variable whose e-power asymptotically matches this upper bound exactly.

\begin{figure}
    \centering
    \includegraphics[width=\linewidth]{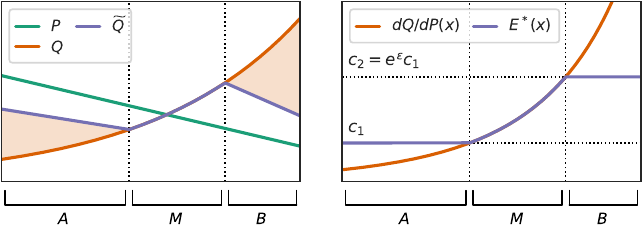}
    \caption{\textbf{Left}: Illustrates our definition of $\widetilde Q$ for $\varepsilon=1$ and two densities $P$ and $Q$ over $[0,1]$. 
    The shaded region in A and the shaded region in B both have area $\TV{\widetilde Q}{\Alt}$.
    \textbf{Right}: Compares the true likelihood ratio $dQ/dP$ to the non-private bounded e-variable $E^*$, which can be interpreted as $d\widetilde Q/dP$. Simply clipping $dQ/dP$ to a fixed range could introduce either positive or negative bias; the distribution-dependent range $[c_1, c_2]$ is calibrated to ensure that  $\E^\Null[E^*]=1$ exactly.}
    \label{fig:construction}
\end{figure}

\subsection{Matching the Optimal Rate}\label{sec:batch_lower_bound}

Our lower bound on the rate proceeds in two steps.
First, we introduce a non-private e-variable $E^*$ with log-sensitivity $\eps$, and show that its e-power matches our upper-bound with equality.
\begin{theorem}\label{thm:exact-lb}
    Let $E^*(x) = \min(c_2, \max(c_1, q(x)/p(x)))$, with $c_1 = e^{-\eps/2 + \lambda^* - 1}$ and $c_2 = e^{\eps}c_1$, as in \cref{sec:batch_construction}.
    Then:
    \begin{equation*}
        \E^\Alt[\log E^*] = \KL{\widetilde Q}{\Null} + \eps \TV{\widetilde Q}{\Alt}.
    \end{equation*}
\end{theorem}

Second, we prove that any e-variable $E$ with log-sensitivity $\eps$ can be efficiently transformed into an $\eps$-DP e-variable whose e-power approaches $\E^\Alt[\log E]$ asymptotically (\cref{thm:basic-batch}). Combining these two results will then yield an $\eps$-DP e-variable satisfying $\lim_{n \to \infty} \frac{1}{n}\E^{\Alt^n}[\log E] = \opt$.

\subsubsection{Privatizing our test statistic}

We have established that the non-private e-value $E^*$ simultaneously attains our target rate and has bounded log-sensitivity $\eps$. We therefore turn to the natural remaining challenge of constructing a private e-variable $\tilde E$ that preserves its asymptotic e-power. 

A seemingly plausible strategy is to simply add zero-mean Laplacian noise $Z \sim \mathrm{Lap}(b)$ to the sum $\sum_{i=1}^n \log E^*(x_i)$ and then exponentiate. Unfortunately, since $\exp(\cdot)$ is strictly convex, this will introduce a positive bias and violate the requirement that $\E^\Null[\tilde E] \leq 1$. 

We might hope to compensate for this by dividing our final output by the quantity $\E[\exp(Z)]$, which can be recognized as the moment-generating function of a Laplacian distribution evaluated at $1$. 
However, the MGF of a Laplace distribution with scale parameter $b$ is equal to $\frac{1}{1-b^2}$, and only exists for scale parameters $b<1$. Meanwhile, satisfying $\eps$-DP by adding Laplacian noise to $\log[E^*]$ directly would require using a scale parameter of $b=1$ exactly. 
This issue was also recognized by~\citet{csillagDifferentiallyPrivateEValues2025a}, but their proposed solution of using Gaussian noise is unavailable under pure DP. 
To overcome this challenge, we define a new, lower-sensitivity test statistic $\Lambda_n(E;\lambda)$, based on a family of merging functions previously studied by \citet{vovkMergingSequentialEvalues2024} for reasons unrelated to privacy:
    \begin{equation}
        \Lambda_n(E;\lambda) = \sum_{t = 1}^n \log (1-\lambda+\lambda  E(x_t)), 
        \qquad 
        \tilde \Lambda_n(E;\lambda) = \Lambda_n(E;\lambda) + Z_b - \log \mathbb{E}[\exp(Z_b)].
        \label{eq:batch_test_statistic}
    \end{equation}

    Here, $Z_b \sim \mathrm{Lap}(b)$ for some scale parameter $b < 1$ which will depend on our chosen value of $\lambda$. The following theorem shows that when $\lambda$ is chosen optimally, \cref{eq:batch_test_statistic} preserves the asymptotic e-power of $E$. We defer the proof to \cref{sec:deferred_batch}.

\begin{theorem}\label{thm:basic-batch}
   Let $E$ be an e-variable for $\Null$ with log-sensitivity $\eps$ such that $\E^\Alt[\log E] = \mu$. Then, for every $\eps = O(1)$, there exist computable values $\lambda \in (0,1)$ and $b<1$ such that $\exp(\tilde \Lambda_n(E;\lambda))$ is an $\eps$-DP e-variable for $\Null$ and:
    \begin{align*}
        \E^\Alt[\tilde\Lambda_n(E;\lambda)] \geq n\mu - \log(n \mu) - O(1).
    \end{align*}
\end{theorem}

\begin{corollary}\label{cor:optimal-stat}
    By \cref{thm:exact-lb}, instantiating \cref{eq:batch_test_statistic} with $E^*$ yields an $\eps$-DP e-variable such that $\lim_{n\to\infty} \frac{1}{n}\E^\Alt[\log E] \geq \KL{\widetilde Q}{\Null} + \eps \TV{\widetilde Q}{\Alt}$.
\end{corollary}

Summarizing our results in the batch setting, we have derived the following chain of inequalities:
\begin{align*}
    \opt 
    &= \lim_{n \to \infty} \sup_{\M\,\eps\text{-DP}} \frac{\KL{\M(Q^n)}{\M(P^n)}}{n} &(\cref{prop:characterization}) \\
    &\leq \inf_{Q' \in \mathcal{D}} \KL{Q'}{\Null} + \eps \TV{Q'}{\Alt} & (\cref{lem:decomp}) \\
    &= \KL{\widetilde Q}{\Null} + \eps \TV{\widetilde Q}{\Alt} & (\cref{thm:e-power-ub}) \\
    &= \E^\Alt[\log E^*] & (\cref{thm:exact-lb}) \\
    &\leq \lim_{n\to\infty} \frac{1}{n}\E^\Alt[\tilde \Lambda_n(E^*; \lambda)] & (\cref{thm:basic-batch}) \\
    &\leq \opt. & (\textit{by definition of optimality})
\end{align*}

Since the first and last terms agree, the inequalities become equalities and \cref{thm:duality} follows as an immediate corollary. Moreover, we have shown that the private e-variable described in \cref{cor:optimal-stat} is asymptotically log-optimal among all $\eps$-DP e-variables.
In \cref{app:constants}, we describe how one can further improve its e-power for finite sample sizes through distribution-dependent post-processing.

Finally, we note that the construction of $E^*$ requires the computation of integrals over the full support of $\Alt$ and $\Null$, which may be expensive or even impossible in some settings (e.g. if $\Alt$ and $\Null$ are very high-dimensional, or defined implicitly through a generative model). In \cref{app:distribution-independent}, we define a distribution-independent alternative to $E^*$, the truncated and scaled likelihood-ratio ($\tsllr_\eps$), and show that it can be used to construct an $\eps$-DP e-variable with asymptotic rate $\Theta(\opt)$ for any $\eps = \Theta(1)$. We leave as an open question whether any distribution-independent $\eps$-DP e-variable can match the optimal rate when $\eps = o(1)$.

\section{Sequential Setting}\label{sec:sequential}

In sequential hypothesis testing, new data is collected gradually over time until some condition is met and the experiment is concluded. The number of data points used to make any final decision is therefore a random variable, $N$, called a \emph{stopping time}. It is typical for $\E[N]$ to be much smaller than the minimum $n$ required in the batch setting~\citep{wald1948optimum}; in fact, the Sequential Probability Ratio Test (SPRT) \citep{wald1945} was first developed during World War II in response to the US Navy's need for efficient alternatives to the likelihood ratio test. Much like the likelihood ratio test, however, the SPRT is decision-theoretic and requires the analyst to commit to a stopping condition in advance.

The analogue in our setting is the \emph{e-process}, which is a sequence of e-variables $E_1, E_2, \ldots$ with the additional property that $E_N$ is itself an e-variable for any stopping time $N$. This leads to some very useful properties, such as optional stopping (we may stop the experiment at any time and draw valid conclusions from our final output) as well as optional continuation (if our results are inconclusive, a second scientist can run a follow-up experiment and merge the resulting e-values by multiplication).  

The theory of e-processes is intimately connected with that of (non-negative) martingales, which have an interesting game-theoretic interpretation as the accumulated wealth of a gambler betting against the null~\citep{shafer2021testing}. We collect several standard definitions and results about martingales in \cref{app:lemmas-and-defs}. \citet[Chapter 12]{upfal2005probability} provide an accessible introduction to the tools we require. 

\begin{definition}
    A process $(M_t)$ is called a \emph{martingale} under $\Null$ if $\E^\Null[M_t \mid M_{1:t-1}] = M_{t-1}$ for all $t$. 
    If the equals sign is replaced by $\leq$ (resp. $\geq$), then it is called a \emph{supermartingale} (resp. \emph{submartingale}).
\end{definition}

\begin{definition}
    A process $(M_t)$ is called a \emph{test (super)martingale for $\Null$} if (a) it is $\Null$-almost surely nonnegative, (b) $M$ is a (super)martingale under $\Null$, and (c) $\mathbb{E}^\Null[M_0] \leq 1$.
\end{definition}

\begin{definition}
    A sequence of e-values $(E_t)$ 
    is called an \emph{e-process for $\Null$} if there exists a test (super)martingale $M$ such that $E \leq M$ $\Null$-almost surely.
\end{definition}

Any test supermartingale is an e-process, but the converse is not true in general. By Ville's inequality, any e-process can be used to construct a level-$\alpha$ sequential test by thresholding at $1/\alpha$.

As in the batch setting, the precise definition of a private e-process requires some care. Essentially, privacy must be defined with respect to the sequence of data points that our algorithm receives as input, whereas e-process properties must be established with respect to the private output sequence.

\begin{definition}[Differentially Private E-process]\label{def:dp-eproc}
    Let $(x_t) \in \mathcal{X}^\mathbb{N}$ be a sequence of data points. For any mechanism $\M : \mathcal{X}^\mathbb{N} \to \mathcal{O}^\mathbb{N}$, define the output process $O_t = \M(x)_t$. 
    We say that a sequence of random variables $(E_t)$
    is an $\varepsilon$-DP e-process for $\Null$ if \textbf{(a)} $\M$ satisfies $\varepsilon$-DP, \textbf{(b)} $O_t$ depends solely on $x_{1:t}$ for all $t$, and \textbf{(c)} $(E_t)$ is an e-process for $\M(\Null^\mathbb{N})$.
\end{definition}

\subsection{Defining the Right Metric}

The standard metric for quantifying the power of an e-process $M$ is the \emph{asymptotic growth rate of $M$ under $\Alt$}, which is defined by the lower limit $\liminf_{t \to \infty} \log(M_t)/t$ (if it is a constant). In particular, for any pair of distributions $\Alt$ and $\Null$, the non-private likelihood ratio process $M_t = \frac{d\Alt^t}{d\Null^t}(x_{1:t})$ achieves the optimal asymptotic growth rate of $\KL{\Alt}{\Null}$
. A weaker metric, usually applied in the context of composite alternates, is to consider the asymptotic e-power $\lim_{t \to\infty} \E^{\Alt^t}[\log E_t]/t$.

The fact that both metrics are stated in terms of limits is justified in the non-private setting, where the asymptotic behavior of likelihood ratio processes and their extensions is closely related to their local behavior for finite samples and increments. But in our case, the addition of privacy constraints introduces a tradeoff: by batching many points together before producing a new output, we can minimize the total error from noise at the cost of additional latency. Metrics defined in terms of limits essentially disregard one half of this tradeoff and therefore incentivize clearly impractical algorithm designs, e.g., involving outrageously large batch sizes.

For this reason, we argue that DP e-processes should be evaluated in terms of their expected value at finite stopping times. 
Moreover, because one of the main advantages of sequential hypothesis testing with e-processes is that we can choose whether to gather more data adaptively without hurting the statistical validity of our conclusions, it would be ideal to have a single private algorithm which is able to perform well at all stopping times simultaneously. In the following section, we prove \cref{thm:agr-ub}, which characterizes the best guarantee of this form that we could hope for.

\subsection{Lower Bound on Expected Stopping Times}

      \begin{theorem}\label{thm:agr-ub}
          Let $(E_t)$ be an $\varepsilon$-DP e-process. Then, for any stopping time $N$ which is bounded above by some constant $N_{max} <\infty$,
        \begin{equation*}
            \E^\Alt\left[\log E_N \right] \leq \E^\Alt[N] \opt.
        \end{equation*} 
      \end{theorem}

We remark that results with a similar form to \cref{thm:agr-ub} are widely known to apply to random sums of independent random variables by Wald's equation. In particular, an analogous result for local DP e-variables follows immediately from classical techniques. The picture is much more complicated in the central DP setting, and the main technical novelty in the proof of \cref{thm:agr-ub} is the introduction of an appropriate potential function which essentially allows us to analyze an arbitrary central DP mechanism as if it operated in a pointwise fashion.

As an application of \cref{thm:agr-ub}, we derive lower bounds on the expected stopping time of \emph{any} $\eps$-DP sequential hypothesis test, including those that are not (obviously) based on e-processes. We define an $\eps$-DP sequential test as an $\eps$-DP mechanism $\M:\mathcal{X}^* \to \lbrace 0, 1\rbrace^*$ with output sequence $\phi_1, \phi_2,\ldots$ and stopping time $N$. We interpret $\phi_t=1$ to mean that $\M$ has rejected the null at or before time $t$. The power of $\M$ is $1-\beta\coloneq \Alt(\phi_N = 1)$, and its level is $\alpha \coloneq \Null(\phi_N=1)$. \cref{thm:agr-ub} immediately implies a lower bound on the stopping time of $\eps$-DP sequential tests that are specifically based on thresholding e-processes; the following proposition shows that this lower bound is in fact universal.
\begin{proposition}\label{prop:sequential-tests}
    Let $\M$ be any $\varepsilon$-DP level-$\alpha$ sequential test with power $1-\beta$ and stopping time $N < N_{max}$ for some constant $N_{max} < \infty$. Then:
    \begin{equation*}
        \E^\Alt[N] \geq \frac{(1-\beta)\log((1-\beta)/\alpha) + \beta\log(\beta/(1-\alpha))}{\opt}.
    \end{equation*}
\end{proposition}

\subsection{Upper Bound on Expected Stopping Times}

\begin{algorithm}
\SetAlgoNoEnd
\SetAlgoNoLine
\caption{Batch to E-process Conversion}\label{alg:batch2proc}
\KwData{Privacy parameter $\varepsilon > 0$; data points $x_1, x_2, \ldots$; e-variable $E(\cdot)$ with log-sensitivity $c\varepsilon$ such that $\E^{\Alt^n}[\log E(X)] = n\mu$ for all $n\ge 1$; competitive ratio $\rho > c$}
Choose $\lambda \in (1/\rho, \min(1,1/c))$ to minimize $t_1$, defined below\;
Compute Laplacian compensator $C_\lambda \gets -\log(1-c^2\lambda^2)$\;
Set minimum stopping time $t_1 \gets \rho\lambda + \frac{\rho^2 \lambda C_\lambda}{\mu(\rho \lambda-1)^2}$\;
For $j \geq 1$, recursively define $t_{j+1} \gets \rho\left( \lambda t_j - \frac{jC_\lambda}{\mu}\right)$\;
Initialize $\widetilde E_0 = 1$\;
\For{$i=1,2,\ldots$}{
  \eIf{$i = \lfloor t_j \rfloor$ for some $j$}{
    Output $\widetilde E_i = \widetilde E_{i-1} \cdot\exp\left(\lambda \log E(x_{t_{j-1}+1:t_j}) + \text{Lap}(\lambda c) - C_\lambda\right)$\;
  }{
      Output $\widetilde E_i = \widetilde E_{i-1}$\;
  }
}
\end{algorithm}

In this section, we provide an efficient algorithm for reducing the sequential setting to the batch setting (\cref{alg:batch2proc}), and we prove that this algorithm is simultaneously optimal up to a constant competitive ratio for all stopping times that are not too small (\cref{thm:sequential-lb}). By tuning the hyperparameters of \cref{alg:batch2proc}, we can trade off between a lower minimum stopping time and a tighter competitive ratio after that initial time is reached.

As in the batch setting, we begin by considering the problem of directly privatizing an existing e-process through additive noise:

\begin{lemma}\label{lem:ind-martingale}
    Let $(M_t)$ be any non-private test (super)martingale such that $M_t$ depends only on $x_{1:t}$ for all $t$, and define $\Lambda_t = \log(M_t)$. Let $\M: \mathbb{R}^\mathbb{N} \to \mathbb{R}^\mathbb{N}$ be an $\eps$-DP mechanism, and let $\tilde \Lambda_t = \M(\Lambda)_t$ depend only on $\Lambda_{1:t}$. Then, if $\lambda \leq 1$ and $\xi_t \coloneq \tilde \Lambda_t - \Lambda_t$ is independent of $\Lambda_t$, the following sequence of random variables is an $\eps$-DP e-process:
    \begin{equation*}
    E_0 = 1,
    \quad ~E_t \coloneq \exp\left( \lambda \tilde \Lambda_t - \sum_{i=1}^t K_i(\lambda) \right), 
    \quad K_i(\lambda) \coloneq \log \mathbb{E}^\Null[\exp(\lambda (\xi_i - \xi_{i-1})) \mid \tilde \Lambda_{1:i-1}].
    \end{equation*}

\end{lemma}

The key quantity in \cref{lem:ind-martingale} is the compensator term $K_i$, which imposes a tradeoff between power and latency; by Jensen's inequality, releasing an unbiased estimate of $\Lambda_t$ induces an upward bias after exponentiating, and the accumulated compensator terms are the price we must pay to maintain the martingale property. This naturally suggests an algorithmic approach based on batching, because re-releasing our last output preserves the martingale property with no long term cost. 

Building on this idea, we present \cref{alg:batch2proc}, a universal mechanism for constructing $\varepsilon$-DP e-processes out of e-variables with bounded sensitivity. Compared to the batching-based algorithm of \citet{csillagDifferentiallyPrivateEValues2025a}, our main technical innovation is the design of a carefully calibrated batching schedule which allows us to optimize e-power over all stopping times simultaneously. The following theorem summarizes the utility guarantees of our algorithm:

\begin{theorem}\label{thm:sequential-lb}
    Let $E$ be a non-private e-variable with log-sensitivity $c\eps$ such that $\E^{\Alt^n}[\log E(X)] = n\mu$ for all $n$, and let $(\widetilde E_t)$ be the output of \cref{alg:batch2proc} with parameters $E, \eps$, and $\rho = \eta^2c$ for some $\eta>1$. Then $(\widetilde E_t)$ is a $\eps$-DP e-process for $\Null$. Moreover, for any stopping time $N \geq \inf_{\lambda \in (1/\rho,1/c)} \left(\rho\lambda - \frac{\rho^2 \lambda \log(1-c^2\lambda^2)}{\mu(\rho \lambda-1)^2}\right) = \widetilde O_{\eta\to 1^+}\left(\frac{c}{\mu(\eta-1)^2}\right)$ which is $\Alt$-almost surely finite, $\E^\Alt[\log \widetilde E_N] \geq \E^\Alt[N]\mu/\rho$.
\end{theorem}

\begin{corollary}
    Instantiating \cref{alg:batch2proc} with $\rho>1$ and the optimal bounded e-variable introduced in \cref{thm:exact-lb} yields an $\eps$-DP e-process satisfying $\E^\Alt[\log \widetilde E_N] \geq \frac{1}{\rho}\E^\Alt[N] \opt$ for any sufficiently large stopping time $N$, matching \cref{thm:agr-ub} up to an arbitrarily small constant factor.
\end{corollary}

\section{Simulations}\label{sec:simulation}

In this section, we compare our \cref{alg:batch2proc} against a private variant of the sequential probability ratio test (SPRT)~\citep{wald1945}. 
The non-private SPRT accumulates the log-likelihood ratio \(L_t=\sum_{i=1}^t \log(q(X_i)/p(X_i))\) and stops when \(L_t\) crosses upper or lower boundaries, which are selected to obtain type I and II errors of $\alpha,\beta$. 
Existing DP sequential testing methods implement variants of the \texttt{AboveThreshold} algorithm~\citep{dwork2014algorithmic} to privately check if $L_t$ has crossed the boundaries. 
We compare against the recently proposed DP-SPRT of \citet{michel2026dpsprtdifferentiallyprivatesequential}, which is restricted to testing single-parameter exponential families with bounded support.
(The other natural baseline of \citet{zhang2022DPSPRT} only satisfies approximate DP.)

\begin{figure}[t]
    \centering
    \includegraphics[width=\textwidth]{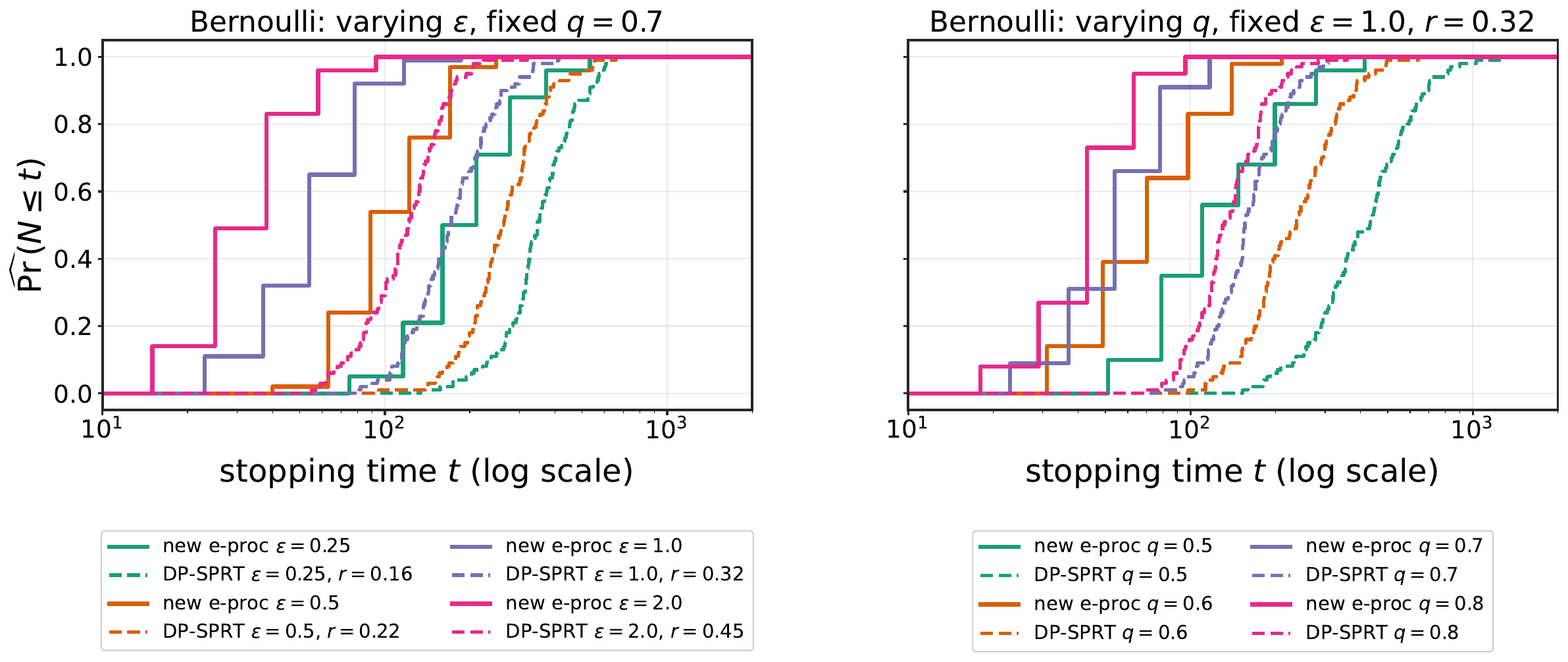}
    \caption{Empirical CDFs of stopping times for sequential tests over 100 trials under $\Alt=\mathrm{Bernoulli}(q)$ with null hypothesis  $\Null=\mathrm{Bernoulli}(0.3)$ and $\alpha=\beta=\frac{1}{40}$. Solid lines represent the two-sided private e-process with $E=E^*$ and $\rho=3$, while dashed lines represent the subsampled DP-SPRT of \cite{michel2026dpsprtdifferentiallyprivatesequential} with their suggested subsampling rate $r(\varepsilon)=\min\{1,\sqrt{\varepsilon/10}\}$. Each solid line sits to the left of the dashed line of the same color, which means that \cref{alg:batch2proc} consistently requires less data than the DP-SPRT to achieve the same statistical power across all values of $q$ and $\eps$.}
    \label{fig:bernoulli-ecdf-subsampled-dpsprt}
\end{figure}

Our evaluation uses the distributions \(\Null=\mathrm{Bernoulli}(0.3)\) and \(\Alt=\mathrm{Bernoulli}(q)\), which satisfy the assumptions of DP-SPRT and were also used in the evaluation of   \citet{michel2026dpsprtdifferentiallyprivatesequential}.  Our e-process construction is one-sided, controlling rejection of a designated null, whereas DP-SPRT is a two-sided test. To obtain a two-sided decision rule, we run two private e-processes in parallel: one for testing $\Null$ against $\Alt$ and one for $\Alt$ against $\Null$. The e-processes use competitive ratio $\rho=3$ and a privacy parameter of $\eps/2$ each. For each method, we use levels \(\alpha = \beta = \frac{1}{40}\) for type I and II error control and report the empirical CDFs of the stopping time $N$ after 100 trials. 

Across the settings shown in \Cref{fig:bernoulli-ecdf-subsampled-dpsprt}, we consistently find that our two-sided private e-process stops 
earlier than DP-SPRT. This is 
despite the fact that our
algorithm makes no parametric assumptions about the distributions being tested, whereas the DP-SPRT was specifically designed for Bernoulli data.
While it is impossible to draw strong conclusions on the basis of a single experiment, 
these results are nevertheless 
a promising sign for the practical usefulness of our algorithms, as well as for DP e-values and e-processes more generally.

\section{Conclusion}

In this work, we investigated simple hypothesis testing with $\varepsilon$-DP e-values. 
In the batch setting, we characterized $\opt$, the optimal instance-specific e-power achievable by any e-value under $\varepsilon$-DP, and constructed a concrete $\varepsilon$-DP algorithm which matches that rate exactly. 
In the sequential setting, 
 we proved a universal lower bound on the expected stopping time of any $\varepsilon$-DP e-process and showed that our \cref{alg:batch2proc} matches that lower bound up to an arbitrarily small constant factor. 
Moreover, we found empirically that our algorithm is able to consistently outperform an existing state-of-the-art algorithm for DP sequential testing.
 In future work, we hope to extend our results to other natural settings such as composite hypothesis testing with DP e-values, as well as relaxed notions of privacy such as approximate DP where the optimal achievable e-power remains unknown.

\bibliography{E-Values_v2.bib}


\appendix

\section{Useful technical lemmas and definitions}\label{app:lemmas-and-defs}

In this appendix, we collect a number of standard results. 
For references and much more background, see \cite{upfal2005probability,ramdasHypothesisTestingEvalues,polyanskiy2025information}.

We will make frequent use of the ``group privacy'' property of DP, which allows the definition of DP to be extended inductively to datasets which differ by more than one record.
\begin{lemma}[Group Privacy]\label{lem:group-privacy}
    Let $\M: \mathcal{X}^* \to \mathcal{O}$ be $\varepsilon$-DP, and let $x, x'$ be two datasets that differ on at most $k$ records. Then, for all $o \in \mathcal{O}$, $\frac{\Pr[\M(x) = o]}{\Pr[\M(x') = o]} \leq e^{k\varepsilon}$.
\end{lemma}

The next several definitions and lemmas describe properties of KL divergence and couplings between probability distributions, which are used extensively in the proofs of our hardness results (\cref{thm:e-power-ub} and \cref{thm:agr-ub}).

\begin{definition}[KL Divergence]
    Given two probability measures $Q \ll P$ on a measurable space $\mathcal{X}$, the KL divergence of $Q$ from $P$ is defined to be:
    \begin{equation*}
        \KL{Q}{P} = \int_{x \in \mathcal{X}} \log \frac{dQ(x)}{dP(x)}~dQ(x),
    \end{equation*}

    where $\frac{dQ(x)}{dP(x)}$ is the Radon--Nikodym derivative of $Q$ with respect to $P$.
\end{definition}

\begin{definition}[Conditional KL Divergence]
    Given two joint probability measures $Q_{XY}$ and $P_{XY}$ with marginals $Q_X$ and $P_X$, the conditional KL divergence of $Q$ from $P$ given $X$ is defined by:
    \begin{equation*}
        \KL{Q_Y(\cdot\mid X)}{P_Y(\cdot \mid X)} = \int_{x \in \mathcal{X}} \KL{Q_{Y\mid X =x}}{P_{Y\mid X=x}} ~dQ_X(x),
    \end{equation*}
    or, equivalently, as $\E^{Q_X}[\KL{Q_{Y|X=x}}{P_{Y\mid X=x}}]$.
    
\end{definition}

\begin{lemma}[Chain Rule for KL Divergence]\label{lem:kl_chain_rule}
    For any pair of joint probability measures $Q_{XY}$ and $P_{XY}$,
    \begin{equation*}
        \KL{Q_{XY}}{P_{XY}} = \KL{Q_X}{P_X} + \KL{Q_Y(\cdot\mid X)}{P_Y(\cdot \mid X)}.
    \end{equation*}
\end{lemma}

The next lemma is stated for simple mixture distributions with only two components, but can be extended inductively to more complex mixtures.
\begin{lemma}[Joint Convexity of KL Divergence]\label{lem:kl_joint_convexity}
    Let $w \in [0,1]$, and let $Q = (1-w)Q' + wQ''$, $P = (1-w)P' + wP''$ be two mixture distributions with matching weights. Then:
    \begin{equation*}
        \KL{Q}{P} \leq (1-w)\KL{Q'}{P'} + w\KL{Q''}{P''}.
    \end{equation*}
\end{lemma}

\begin{lemma}[Donsker--Varadhan Lemma]\label{lem:donsker_varadhan}
    Let $Q$ and $P$ be two probability measures on $\mathcal{X}$, and let $\mathcal{F}_P$ be the set of all measurable functions $f$ such that $\E^P[e^{f(x)}] < \infty$. If $Q \ll P$, then:
    \begin{equation*}
        \KL{Q}{P} = \sup_{f \in \mathcal{F}_P} \left\lbrace \E^Q[f(x)] - \log \E^P[e^{f(x)}] \right\rbrace.
    \end{equation*}
\end{lemma}

\begin{definition}[Couplings]
    Given two probability measures $Q$ and $P$ on $\mathcal{X}$, a coupling of $Q$ and $P$ is a joint probability measure $\gamma$ on $\mathcal{X} \otimes \mathcal{X}$ such that $\gamma(A, \mathcal{X}) = Q(A)$ and $\gamma(\mathcal{X}, B) = P(B)$ for all measurable events $A, B$.
\end{definition}

\begin{lemma}[Existence of Maximum Couplings]\label{lem:tv-coupling}
    For any pair of probability measures $Q$ and $P$, there exists a maximal coupling $\gamma$ such that $\gamma(X \neq Y) = \TV{Q}{P} = \sup_A|Q(A) - P(A)|$. 
\end{lemma}

The following definitions and lemmas describe fundamental properties of martingales and e-processes, which are used throughout \cref{sec:sequential}.

\begin{definition}[Martingales]
    Let $(X_t)$ be a sequence of random variables adapted to some filtration $\mathcal{F}$. We say that $(X_t)$ is a martingale under $P$ if $\E^P[|X_i|] < \infty$ for all $i \geq 1$ and
    \begin{equation*}
        \E^P[X_t \mid \mathcal{F}_{t-1}] = X_{t-1},
    \end{equation*}
    i.e. the conditional expectation of the next value given the information available is always equal to the current value. If the equality in the above definition is replaced by $\leq$ (resp. $\geq$), then we say that $(X_t)$ is a supermartingale (resp. submartingale).
\end{definition}
 
The next theorem can be stated with considerably more general assumptions. We present the simplest version that suffices for our purposes. 

\begin{lemma}[Optional Stopping Theorem]\label{lem:optional_stopping}
    Let $(M_t)$ be a martingale under $P$ adapted to the filtration $\mathcal{F}$, and let $N$ be a stopping time for $\mathcal{F}$. If there is some positive integer $N_{max}$ such that $N \leq N_{max}$ almost surely, then $\E^P[M_N] = \E^P[M_1]$.
\end{lemma}

\begin{lemma}[Wald's Equation]\label{lem:wald}
    Let $(X_t)$ be an infinite sequence of random variables with $\E^P[X_i]=\mu$ for some distribution $P$, and let $N$ be a non-negative integer-valued random variable such that $X_n$ is independent of the event $\lbrace N \geq n\rbrace$. Then, provided that the infinite series satisfies $\sum_{n=1}^\infty \E^P[|X_n| \mathbb{I}\lbrace N \geq n \rbrace] < \infty$,
    \begin{equation*}
        \E^P\left[\sum_{n=1}^N X_n \right] = \E^P[N]\mu.
    \end{equation*}
\end{lemma}

\begin{lemma}[Ville's Inequality]
    If $(M_t)$ is an e-process for $\Null$, then for every $\alpha \in [0,1]$,
    \begin{equation*}
        \Null\left( \exists t: M_t \geq \frac{1}{\alpha}\right) \leq \alpha.
    \end{equation*}
\end{lemma}

\section{Improving our private e-variable for finite sample sizes}\label{app:constants}

The private e-variable introduced in \cref{eq:batch_test_statistic} is asymptotically log-optimal. This does not imply that \cref{eq:batch_test_statistic} cannot be improved for specific input distributions, however. In particular, because we directly add Laplacian noise, our private e-variable is always unbounded, even when the non-private likelihood ratios themselves are not. As we show below, this implies that the construction is suboptimal in a fairly strong sense.

\begin{proposition}[Suboptimality of unbounded private e-variables]\label{prop:suboptimality}
    If $E$ is an unbounded $\varepsilon$-DP e-variable, then there exists a function $f$ such that $f(E)$ is an $\eps$-DP e-variable and $\E^\Alt[\log(\frac{E}{f(E)})] < 0$.
\end{proposition}

\begin{proof}    
    Since $E$ is an $\eps$-DP e-variable, we can write it as $\M(X)$ for some $\eps$-DP mechanism $\M$. Then, by \cref{lem:dp-ratio}, the unique log-optimal e-variable for testing $\M(\Alt^n)$ against $\M(\Null^n)$ is the likelihood ratio $E_\M(X)$, and by group privacy, $E_\M(X)$ is bounded in $[e^{-n\varepsilon}, e^{n\varepsilon}]$. Taking contrapositive, we conclude that no unbounded $\varepsilon$-DP e-variable can be log-optimal, even when compared against mechanisms that operate solely by post-processing its outputs.
\end{proof}

This observation naturally suggests the following modification of our original construction, which is log-optimal for its own output distribution by design:
\begin{align}
    \M(X) =  \sum_{t = 1}^n \log (1-\lambda+\lambda  E^*(x_t)) + Z_b,  \qquad E = \frac{d\M(\Alt^n)}{d\M(\Null^n)}(\M(X)).
\end{align}

When can we actually compute this value? Under the hypothesis that the $x_i$ are drawn i.i.d.\ from some distribution $D \in \lbrace \Null, \Alt \rbrace$, $\M(x_{1:n})$ is a sum of independent random variables. So, a straightforward sufficient condition is the ability to compute the CDF of the untruncated likelihood ratio under $D$, which is in turn sufficient to compute the characteristic function of $E^*$, $\varphi_D(t) \coloneq \E_{x \sim D}[\exp(it\cdot\E^*(x))]$. In this case, we can recover the likelihood under $D$ as
\begin{equation}
    \mathcal{F}^{-1}\left( \varphi_D(t)^n \cdot \frac{1}{1+b^2t^2} \right)(\M(X)),
\end{equation}

where $\mathcal{F}^{-1}$ denotes the inverse Fourier transform. We remark that, although the exact e-power of this e-variable typically does not have a closed form, log-optimality allows us to lower-bound it by the e-power of any other e-variable derived from the same private output, including the one described in \cref{thm:basic-batch}.

\section{A nearly-optimal distribution-independent bounded e-variable}\label{app:distribution-independent}

In this appendix, we present the truncated, scaled likelihood ratio ($\tsllr_\eps$) statistic, a simple bounded e-variable which does not require any form of distribution-dependent tuning, and show that it achieves e-power of at least $\frac{\eps-1}{\eps}(1-e^{-\eps})\opt$. This makes the $\tsllr_\eps$ statistic a viable alternative to $E^*$ in settings where computing the distribution-dependent tuning parameter $\lambda^*$ is computationally intractable, provided that $\eps$ is at least somewhat larger than 1. For constant $\eps \leq 1$, we can still match the optimal rate up to constant factors by computing $\tsllr_{\eps'}$ at a larger, fixed value of $\eps'$ and then taking fractional powers, which we show yields a bounded e-variable with e-power of at least $\frac{\eps}{5}\opt$. In either case, we can then privatize our bounded e-variable, which we show can reduce its e-power by at most a factor of 2.

\subsection{A second upper-bound on $\opt$}
\label{sec:batch_upper_bound}

To prove the optimality of our second construction, it will be easiest to begin by proving a slightly different upper bound on $\opt$ than the one presented in \cref{sec:batch_rate_upper}.
This upper bound will be based on a decomposition of $\Null$ and $\Alt$ into a pair of mixture distributions with matching weights. We begin by introducing the hockey-stick divergence of $\Alt$ from $\Null$:
\begin{equation*}
    \tau = D_{e^\varepsilon}(\Alt \Vert \Null) =\int_\mathcal{X} 
      \max \lbrace q(x)-e^\varepsilon p(x), 0 \rbrace ~dx.
\end{equation*}

Next, let $B = \lbrace x : q(x) > e^\varepsilon p(x)\rbrace$, so that $\Alt(B) = \tau + e^\varepsilon \Null(B)$. Define a subdistribution $\bar q(x) = \min(q(x), e^\varepsilon p(x))$ with normalized density function $q'(x) = \bar q(x)/(1-\tau)$. This allows us to write $\Alt = (1-\tau)Q' + \tau Q''$, where $Q''$ is some distribution supported on $B$. We will similarly define the (trivial) decomposition $\Null = (1-\tau)P' + \tau P''$, where $P' = P'' = \Null$. With these definitions, we can now present our upper bound:

\begin{theorem}\label{thm:tslr-ub}
        Let $\Null$ be an arbitrary distribution over $\mathcal{X}$, and let $E$ be an $\eps$-DP e-variable for $\Null$. Then for all distributions $\Alt$ over $\mathcal{X}$, we have:
        \begin{equation*}
            \frac{\mathbb{E}^\Alt[\log E]}{n} \leq (1-\tau)\KL{Q'}{P'} + \tau\varepsilon.
        \end{equation*}
\end{theorem}

\begin{proof}
We will use the same techniques as in the proof of \cref{lem:decomp}. 
Let $\mathcal{M}$ be an $\eps$-DP algorithm that takes a dataset $x_{1:n}$ as input and outputs an e-value $E$ for testing $x_i \sim \Alt$ against $x_i \sim \Null$. Let $\M(\Alt)$ (resp. $\M(\Null))$ denote the pushforward measure of the algorithm's output space. By Donsker--Varadhan (\cref{lem:donsker_varadhan}), we have:
\begin{equation*}
    \E^\Alt[\log E] \leq \KL{\M(\Alt)}{\M(\Null)} + \log \E^\Null[E] \leq \KL{\M(\Alt)}{\M(\Null)},
\end{equation*}

where the second inequality follows because $E$ is an e-variable. Using the decomposition described above, define a coupling of $\Null$ and $\Alt$ as follows: for each $t \in [n]$, we sample $b_t \sim Bern(\tau)$. If $b_t=0$, then $x_t \sim Q'$ and $y_t \sim P'$. Otherwise, $x_t \sim Q''$ and $y_t \sim  P''$. Define $\mathcal{I}$ to be the set of indices for which $b_i=0$, and let $x_\mathcal{I} = \lbrace x_i : i \in \mathcal{I} \rbrace$. 

    To make the following calculations clearer, we suppress $\M$ from our notation and let $\Null_{E}$ denote the distribution of $E$ under $\Null$ and let $\Null_{E,x_\cI}$ denote the joint distribution of $E$ and $x_{\cI}$, and let $\Null_{E,x_\cI\mid \cI}$ denote the joint distribution conditioned on (a realization of) $\cI$.
    Define the notation for $\Alt$ likewise.
    Rewriting and applying monotonicity and the joint convexity of KL divergence (\cref{lem:kl_joint_convexity}), we have
    \begin{align*}
        \KL{\M(\Alt)}{\M(\Null)} 
            = \KL{\Alt_E}{\Null_E} 
            &\le \KL{\Alt_{E,x_\cI}}{\Null_{E,x_\cI}} \\
            &\le \E_\cI\left[ \KL{\Alt_{E, x_\cI\mid \cI}}{\Null_{E, x_\cI\mid \cI}}\right].
    \end{align*}
    We then apply the chain rule for KL divergence (\cref{lem:kl_chain_rule}) and obtain
    \begin{align*}
        \KL{\M(\Alt)}{\M(\Null)} \le \E_\cI\left[ \KL{\Alt_{x_\cI\mid \cI}}{\Null_{x_\cI\mid \cI}}\right]
            +\E_{\cI,x_\cI}\left[ \KL{\Alt_{E\mid \cI, x_\cI}}{\Null_{E\mid \cI,x_\cI}}\right].
    \end{align*}
    For the first term above, by construction points in $x_\cI$ are drawn i.i.d.\ from $Q'$ or $P'$, so we have
    \[
        \E_\cI\left[ \KL{\Alt_{x_\cI\mid \cI}}{\Null_{x_\cI\mid \cI}}\right] = \E_\cI\left[|\cI|\cdot \KL{Q'}{P'}\right].
    \]
    For the second term we apply group privacy: since we condition on $x_\cI$ the datasets differ on at most $n-|\cI|$ points and thus
    \[
        \E_{\cI,x_\cI}\left[ \KL{\Alt_{E\mid \cI, x_\cI}}{\Null_{E\mid \cI,x_\cI}}\right]
            \le \E_{\cI}\left[ (n - |\cI|)\eps\right].
    \]
    Applying $\E[|\cI|] = n(1-\tau)$ and dividing by $n$ finishes the proof.
\end{proof}

\subsection{Constructing a distribution-independent test statistic}

We begin by designing a bounded test statistic that achieves $\E^\Alt[\log E] = \Theta(\opt)$ as long as $\eps$ is somewhat larger than 1. Then, we will extend to $\eps=\Theta(1)$ by transforming $E$ in a way which reduces its log-sensitivity while preserving its e-power up to a constant factor. 

\begin{lemma}Define the truncation operator $f_c(x) = \min(c,x)$.  Fix $\Null$ and $\Alt$ and let $R = \frac{d\Alt}{d\Null}(x)$. Then $\E^\Alt[\log f_{e^\varepsilon}(R)] \geq \frac{\eps-1}{\eps}\opt$.
\end{lemma}

\begin{proof}
As in the proof of the upper-bound, define the sub-distribution $\bar q(x) = \min(q(x), e^\varepsilon p(x))$ and let $q'(x) = \bar q(x)/(1-\tau)$ be the normalized density function. This allows us to write $Q = (1-\tau)Q' + \tau Q''$. Let $B = \lbrace x: \log(q(x)/p(x)) > \varepsilon \rbrace$. Now, have that:
\begin{align*}
    \E^\Alt[\log f_{e^\varepsilon}(R)] &= \int_\mathcal{X} q(x) \log(f_{e^\varepsilon}(R))~dx \\
    &= \varepsilon \int_B q(x)~dx + \int_{\mathcal{X} \setminus B} q(x) \log \left( \frac{q(x)}{p(x)} \right)~dx \\
    &= \varepsilon e^\varepsilon \Null(B) + \varepsilon \tau + \int_{\mathcal{X} \setminus B} q(x) \log \left( \frac{q(x)}{p(x)} \right)~dx.
\end{align*}

Next, we compute $\KL{Q'}{P'}$:

\begin{align*}
    (1-\tau)\KL{Q'}{P'}
    &= \int_\mathcal{X} \bar q(x) \log\left( \frac{q'(x)}{(1-\tau)p(x)} \right) ~dx \\
    &= \int_B e^\varepsilon p(x) \log\left( \frac{e^\varepsilon p(x)}{ p(x)} \right) ~dx + \int_{\mathcal{X} \setminus B} q(x) \log\left( \frac{q(x)}{p(x)} \right) ~dx - (1-\tau)\log(1-\tau)\\
    &= \varepsilon e^\varepsilon \Null(B) + \int_{\mathcal{X} \setminus B} q(x) \log\left( \frac{q(x)}{p(x)} \right) ~dx - (1-\tau)\log(1-\tau).
\end{align*}

It follows that:
\begin{align*}
    \frac{\opt}{\E^\Alt[\log f_{e^\eps}(R)]} 
    &\leq \frac{(1-\tau) \KL{Q'}{P'}) + \tau\varepsilon}{\E^\Alt[\log f_{e^\eps}(R)]} \\
    &= \frac{(1-\tau) \KL{Q'}{P'}) + \tau\varepsilon}{(1-\tau) \KL{Q'}{P'}) + \tau\varepsilon + (1-\tau)\log(1-\tau)} \\
    &\leq \frac{\tau\varepsilon}{\tau\varepsilon + (1-\tau)\log(1-\tau)} \\
    &\leq \frac{\tau \eps}{\tau\eps - \tau} \\
    &= \frac{\eps}{\eps-1}
\end{align*}

and thus that $\E^\Alt[\log f_{e^\eps}(R)] \geq \frac{\eps-1}{\eps}\opt$.
\end{proof}

For $\varepsilon>0$, we define the \emph{truncated scaled likelihood ratio statistic}:
\begin{equation*}
    \tsllr_\varepsilon(x) = e^{-\varepsilon} + (1-e^{-\varepsilon})\left( f_{1+e^\varepsilon}\left(\frac{d\Alt}{d\Null}(x)\right) \right).
\end{equation*}

\begin{theorem}\label{thm:tslr}
    The $\tsllr$ statistic satisfies the following properties:
    \begin{itemize}
        \item $\tsllr_\eps$ is an e-variable for $\Null$.
        \item For all $x$, $\log\tsllr_\eps(x) \in [-\varepsilon, \varepsilon]$.
        \item $\E^\Alt[\log \tsllr_\eps(x)] \geq  \frac{(\eps-1)(1-e^{-\varepsilon})}{\eps}\opt$.
    \end{itemize}
\end{theorem}

\begin{proof}
First, we have $\E^\Null[\tsllr_\eps(x)] =e^{-\varepsilon} + (1-e^{-\varepsilon})\E^\Null[f_{1+e^\varepsilon}(\frac{d\Alt}{d\Null}(x))] \leq 1$ because $f_{1+e^\varepsilon}(\frac{d\Alt}{d\Null}(x))$ is an e-variable.
    
Next, we show that $\log\tsllr_\eps$ has bounded sensitivity. Clearly, $\log \tsllr_\eps(x) \geq -\varepsilon$. Simultaneously, we have $\log \tsllr_\eps(x) \leq \log(e^{-\varepsilon} + (1-e^{-\varepsilon})(1+e^\varepsilon)) = \log(e^\varepsilon)=\varepsilon$. So, our statistic has bounded log-sensitivity of $2\varepsilon$.

Finally, we show that the $\tsllr_\eps$ has nearly-optimal e-power for sufficiently large $\eps$. We will bound the e-power of our statistic against $f_\varepsilon$ on the three regions $A = \lbrace x: p(x)>q(x) \rbrace$, $B = \lbrace x:q(x) > (1+e^\varepsilon) p(x) \rbrace$, and $M = \mathcal{X} \setminus (A \cup B)$.

On $A$, we have $\tsllr_\eps(x) = e^{-\varepsilon} + (1-e^{-\varepsilon})\frac{q(x)}{p(x)} \geq \frac{q(x)}{p(x)}$ because $\frac{q(x)}{p(x)} \leq 1$, and so $\log \tsllr_\eps(x) \geq \log f_{e^\varepsilon}(R)$. On $B$, we have $\log \tsllr_\eps(x) = \varepsilon = \log f_{e^\varepsilon(R)}$. Finally, on $M$, we have $\log (\tsllr_\eps(x)) \geq (1-e^{-\varepsilon}) \log \frac{q(x)}{p(x)} \geq (1-e^{-\varepsilon})\log f_{e^{\varepsilon}}(R)$. We conclude that $\log (\tsllr_\eps(x)) \geq (1-e^{-\varepsilon}) \log f_{e^\varepsilon}(R)$ pointwise, and therefore that:
\begin{equation*}
        \E^\Alt[\log \tsllr_\eps(x)] \geq \frac{(\eps-1)(1-e^{-\varepsilon})}{\eps}\opt .
\end{equation*}
\end{proof}

The $\tsllr_\eps$ statistic can be used directly for sufficiently large $\eps$. For example, if $\eps > 2.3$, then one can calculate that $\E^\Alt[\log \tsllr_\eps] \geq \frac{1}{2}\opt$. For $\eps \leq 1$, however, it has no provable guarantees in general. 

\subsection{Extending the distribution-independent test statistic}

We overcome this challenge in two steps. First, we observe that if $E'$ is an e-variable such that $\log E' \in [-\eps', \eps']$ for some $\eps' > \eps$ and $\E^\Alt[\log E'] = \mu$, then $E = \exp(\frac{\eps}{\eps'} \log E')$ is an e-variable (by concavity) which satisfies $\log E \in [-\eps, \eps]$ and $\E^\Alt[\log E] = \frac{\eps}{\eps'}\mu$. Second, we observe that the optimal rate $\opt$ is clearly a non-decreasing function of $\eps$, i.e. $\mathfrak{R}_{\eps'}(\Alt~\Vert~\Null) \geq \opt$. This allows us to construct the following continuous family of candidate bounded-sensitivity e-variables:
\begin{equation*}
    \tsllr_\eps^{\eps'} \coloneq (\tsllr_{\eps'})^{\eps/\eps'} \in [e^{-\eps}, e^\eps],
\end{equation*}

and we can choose $\eps^* \geq \eps$ in order to optimize our lower bound on their e-power:
\begin{equation*}
    \E^\Alt[\log \tsllr_\eps^{\eps^*}] \geq \sup_{\eps' \geq \min(1,\eps)} \left(\frac{(\eps'-1)(1-e^{-{\eps'})}} {\eps'}\right)\left(\frac{\eps}{\eps'}\right) \opt.
\end{equation*}

The right-hand side is a linear function of $\eps$, and therefore the value of the unconstrained optimum is independent of $\eps$. We numerically calculate that the maximum is obtained at $\eps^* \approx 2.334$ where the value of the right-hand side is approximately $0.221 \cdot \eps$. For $\eps > 2.334$, the optimum is attained at the extreme point $\eps^* = \eps$, i.e. the unadjusted $\tsllr_\eps$ statistic. For $\eps = \Theta(1)$, we simplify these constants slightly to conclude that:
\begin{equation*}
    \E^\Alt[\log \tsllr_\eps^{\eps^*}] \asymp \min\left( \frac{\eps}{5},1 \right)\opt = \Theta(\opt)
\end{equation*}

Finally, we prove a slightly modified version of \cref{thm:basic-batch} for e-variables with log-range exactly $[-\eps, \eps]$, which allows us to derive a distribution-independent $\eps$-DP e-variable satisfying 
\begin{equation*}
    \frac{\E^\Alt[\log E]}{n} \asymp \min\left(\frac{\eps}{10}, \frac{1}{2}\right) \opt = \Theta(\opt),
\end{equation*}
as desired.

The formal statement of the modified theorem follows:
\begin{theorem}\label{thm:basic-batch-symmetric}
    Let $E \in [e^{-\eps}, e^\eps]$ be an e-variable for $\Null$ such that $\E^\Alt[\log E] = \mu$. Then, for every $\eps = O(1)$, there exist computable values $\lambda \in (0, 1/2)$ and $b<1$ such that $\exp(\tilde \Lambda_n(E;\lambda))$ is an $\eps$-DP e-variable for $\Null$ and:
    \begin{equation*}
        \E^\Alt[\tilde \Lambda_n(E; \lambda)] \geq \frac{n\mu}{2} - \log(n\mu) - O(1)
    \end{equation*}
\end{theorem}

We provide the (rather technical) proof below. It is conceptually similar to the proof of \cref{thm:basic-batch}, except that the particularly symmetric range of $[e^{-\eps}, e^\eps]$ allows us to derive a more precise closed-form expression for the optimal choice of $\lambda$, which in turn gives a tighter characterization of the final error.

\begin{proof}[Proof of \cref{thm:basic-batch-symmetric}]
Let $E \in [e^{-\eps}, e^\eps]$ be an e-variable for $\Null$. Then
the non-private test statistic from \cref{eq:batch_test_statistic} instantiated with $E$ has log-sensitivity:
    \begin{equation*}
        R_\lambda = \log \left( \frac{1-\lambda+\lambda e^{\varepsilon}}{1-\lambda+\lambda e^{-\varepsilon}} \right).
    \end{equation*}

In order to ensure that we end up with a valid e-variable, we need $\mathbb{E}[e^{Z_b}] = \frac{1}{1-b^2} < \infty$, which implies that $b < 1$. Meanwhile, to satisfy $\varepsilon$-DP, we require that $b \geq \frac{R_\lambda}{\varepsilon}$. Writing $\mu = \E^\Alt[\log E]$ and using the concavity of $\log(\cdot)$, we have that:
    \begin{equation*}
        \mathbb{E}^\Alt[\tilde \Lambda_S(\lambda)] = n \cdot\mathbb{E}^\Alt[\log(1-\lambda+\lambda E)] + \mathbb{E}[Z_b] - \log \mathbb{E}[\exp(Z_b)] \geq \lambda n \mu + \log (1-b^2).
    \end{equation*}

    We will choose $b = \frac{R_\lambda}{\varepsilon}$ exactly and reparameterize with $\lambda = \frac{1}{2} - s$. Then:
    \begin{align*}
        R_s &= \log \left( \frac{1 + s - \frac{1}{2} + (\frac{1}{2}-s)e^\varepsilon}{1 + s - \frac{1}{2} + (\frac{1}{2}-s)e^{-\varepsilon}} \right) \\
        &= \log \left( \frac{1 + 2s + (1-2s)e^\varepsilon}{1 + 2s + (1-2s)e^{-\varepsilon}} \right) \\
        &= \varepsilon+ \log \left( \frac{1 + 2s + (1-2s)e^\varepsilon}{(1 + 2s)e^\varepsilon + 1-2s} \right) \\
        &= \varepsilon+ \log \left( \frac{e^\varepsilon+1 + (2-2e^\varepsilon)s}{(2e^\varepsilon-2)s + e^\varepsilon+1} \right).
    \end{align*}

    Taking derivatives, we get:
    \begin{align*}
        R'_s &= \frac{2-2e^\varepsilon}{e^\varepsilon+1+(2-2e^\varepsilon)s} - \frac{2e^\varepsilon-2}{(2e^\varepsilon-2)s+e^\varepsilon+1} \\
        R''_s &= \frac{-(2-2e^\varepsilon)^2}{(e^\varepsilon+1+(2-2e^\varepsilon)s)^2} + \frac{(2e^\varepsilon-2)^2}{((2e^\varepsilon-2)s+e^\varepsilon+1)^2}.
    \end{align*}

    At $s=0$, we get $R'_s = -4\tanh(\varepsilon/2)$ and $R''_s=0$. Otherwise, we have $R''_s < 0$ whenever:
    \begin{align*}
        \frac{(2-2e^\varepsilon)^2}{(e^\varepsilon+1+(2-2e^\varepsilon)s)^2} 
        &> 
        \frac{(2e^\varepsilon-2)^2}{((2e^\varepsilon-2)s+e^\varepsilon+1)^2} 
        \\
        ((2e^\varepsilon-2)s+e^\varepsilon+1)^2
        &> 
        (e^\varepsilon+1+(2-2e^\varepsilon)s)^2
        \\
        (e^{2\varepsilon}-1)s 
        &>
        (1-e^{2\varepsilon})s \\
        s &> 0.
    \end{align*}

    So, we have that $R_s$ is concave on the interval $[0, 1/2]$, which means that we can uniformly upper-bound it by $R_0 + R'_0s = \varepsilon - 4\tanh(\varepsilon/2)s$. Rewriting our objective function with the same change of variables, we get:

    \begin{equation*}
        \left(\frac{1}{2}-s\right)n\mu + \log\left(1 - \left( 1 - \frac{4\tanh(\varepsilon/2)s}{\varepsilon} \right)^2 \right).
    \end{equation*}

    Let $C = \frac{4\tanh(\varepsilon/2)}{\varepsilon}$. Then the expression simplifies to:
    \begin{equation*}
        \left(\frac{1}{2}-s\right)n\mu + \log\left(1 - \left( 1 - Cs \right)^2 \right).
    \end{equation*}

    Applying the first-order optimality condition with respect to $s$, we get:
    \begin{align*}
        0 &= -n\mu + \frac{2C(1-Cs)}{1-(1-Cs)^2} \\
        0 &= (1-Cs)^2n\mu + 2C(1-Cs) - n\mu \\
        (1-Cs) &= \frac{-C \pm \sqrt{C^2+n^2\mu^2}}{n\mu} \\
        s &= \frac{C - \sqrt{C^2+n^2\mu^2} +n\mu}{Cn\mu}     
        \end{align*}

    At this point, we introduce the change of variables $n\mu = B$, $C = B \sinh(u)$. Then:
    \begin{equation*}
        Cs = \frac{B \sinh(u) - B\cosh(u) + B}{B} = \sinh(u) - \cosh(u) + 1 = 1-e^{-u},
    \end{equation*}
    and so
    \begin{equation*}
        \log(1-(1-Cs)^2) = \log(1-e^{-2u}).
    \end{equation*}
    Simultaneously, we have that:
    \begin{equation*}
        Bs = \frac{B\sinh(u) - B\cosh(u) + B}{B\sinh(u)} = 1+\frac{1-\cosh(u)}{\sinh(u)} = 1 - \tanh(u/2).
    \end{equation*}

    So, all together, the value of our objective function becomes:
    \begin{align*}
        \frac{B}{2} - \left(1-\frac{e^u-1}{e^u+1}\right) + \log(1-e^{-2u})
        &= \frac{B}{2} - \frac{2}{e^u+1} + \log(1-e^{-2u}) \\
        &\geq \frac{B}{2} - 1 + \log(1-e^{-2u}).
    \end{align*}

    Using the logarithmic representation of $u = \sinh^{-1}\left( \frac{C}{B} \right)$, this simplifies to:

    \begin{equation*}
        \frac{B}{2}-1+\ln \left(\frac{2C e^{-u}}{B}\right) = \frac{B}{2} - \ln(B) -1 +\ln \left(2C\right) - u,
    \end{equation*}

    and expanding the definition of $B$ then yields the theorem statement.
\end{proof}

\section{Deferred proofs}
\label{sec:deferred_proofs}

\subsection{Deferred proofs from \cref{sec:batch}}
\label{sec:deferred_batch}

\begin{proof}[Proof of \cref{prop:characterization}]

We begin by proving the following helpful lemma:

\begin{lemma}\label{lem:dp-ratio}
    Let $\M: \mathcal{X}^n \to \mathcal{O}$ be an arbitrary $\varepsilon$-DP mechanism. Then the following likelihood ratio always exists:
    \begin{equation}
        E_\M(X)\coloneq\frac{d\M(\Alt^n)}{d\M(\Null^n)}(\M(X)).
    \end{equation}

    Moreover, it is a $\eps$-DP e-variable for $\Null$, $\E^\Alt[\log E_\M(X)] = \KL{\M(\Alt^n)}{\M(\Null^n)}$, and it is the log-optimal e-variable for testing $\M(\Alt^n)$ against $\M(\Null^n)$.
\end{lemma}
\begin{proof}[Proof of \cref{lem:dp-ratio}]
    To show that the ratio exists, let $A$ be any event with $\M(\Alt^n)(A) > 0$. Then by \cref{lem:group-privacy}, $\M(\Null^n)(A) \geq e^{-n\varepsilon}\M(\Alt^n) > 0$, and so $\M(\Alt^n) \ll \M(\Null^n)$.  Next, $E_\M(X)$ is $\eps$-DP by post-processing, and it is an e-variable because $\E^{\Null^n}[E_\M(X)] = \int_\mathcal{O} (d\M(\Alt^n)/d\M(\Null^n))~d\M(\Null^n) = \int_\mathcal{O} d\M(\Alt^n) = 1$. Finally, $\E^{\Alt^n}[\log E_\M(X)] = \KL{\M(\Alt^n)}{\M(\Null^n)}$ by the definition of KL-divergence, and as a likelihood ratio, $E_\M$ is log-optimal for testing $\M(\Alt^n)$ against $\M(\Null^n)$ by Proposition 3.22 of \citet{ramdasHypothesisTestingEvalues}.
\end{proof}

We proceed with the proof of \cref{prop:characterization}. Let $x_{1:n}$ be data generated i.i.d.\ from either $\Null$ or $\Alt$, and let $E=\M(x_{1:n})$ be an $\eps$-DP e-variable for $\Null$. Then, by the Donsker--Varadhan lemma (\cref{lem:donsker_varadhan}), it follows that:
\begin{align*}
    \E^{\Alt^n}[\log E] &\leq \log(\E^{\Null^n}[E]) +  \KL{\M(\Alt^n)}{\M(\Null^n)}\\
    &\leq \KL{\M(\Alt^n)}{\M(\Null^n))},
\end{align*}
where the second inequality follows because $\M(x_{1:n})$ is an e-variable for $\Null$. Hence,
\[\opt \leq \sup_{\M~\eps\text{-DP}} \frac{\KL{\M(\Alt^n)}{\M(\Null^n)}}{n}.\]

But, for any arbitrary $\eps$-DP mechanism $\M$, \cref{lem:dp-ratio} tells us that the right-hand side is attained exactly by the $\eps$-DP e-variable $E_\M(X)$, and we conclude that the inequality is in fact an equality. 
\end{proof}

\vspace{1cm}

In \cref{thm:exact-lb}, we show that the e-power of the non-private e-variable $E^*$ exactly matches the upper bound given by \cref{thm:e-power-ub} under the minor technical condition that the set $\mathcal{X}$ is a Polish space, which encompasses essentially all interesting cases.

\begin{proof}
    Let $\mathcal{X}$ be a Polish space, and let $\mathcal{D}$ denote the set of probability distributions over $\mathcal{X}$ endowed with the weak topology. Define:
    \begin{equation*}
        f(Q) \coloneq \KL{Q}{\Null} + \eps \TV{Q}{\Alt}, \qquad \widetilde Q \coloneq \argmin_{Q \in \mathcal{D}} f(Q).
    \end{equation*}

    We will begin by showing that $\widetilde Q$ exists and is almost-surely unique. To this end, we observe that $f(\Null) = \eps \TV{\Null}{\Alt}$ is feasible. So, it suffices to optimize over the set $S \coloneq\lbrace Q \mid \KL{Q}{\Null} \leq \eps \TV{\Null}{\Alt}\rbrace$, which is weakly compact because $\mathcal{X}$ is a Polish space. Then, because $f$ is weakly lower semi-continuous, it follows from the extreme value theorem that a minimizer of $f$ exists in $S$. Finally, because KL divergence is strictly convex over $S$ while TV distance is convex, $f$ is strictly convex and we conclude that the minimizer is unique.

    Having established existence and uniqueness, we now turn to analyzing the specific form of the optimum. For a fixed distribution $R$ with density $r$, we define $g(x) = q(x)/p(x)$ and $s(x) = r(x)/p(x)$. Then we can expand the objective function as:
    \begin{align*}
        f(R) &= \int_{\mathcal{X}} r(x) \log \left( \frac{r(x)}{p(x)} \right)~dx + \frac{\eps}{2}\int_{\mathcal{X}} |r(x) - q(x)|~dx\\ 
        &= \int_{\mathcal{X}} \left[\frac{r(x)}{p(x)} \log \left( \frac{r(x)}{p(x)} \right) + \frac{\eps}{2}\left| \frac{r(x)}{p(x)} - \frac{q(x)}{p(x)} \right| \right]p(x)~dx \\
        &= \int_{\mathcal{X}} \left[s(x) \log \left( s(x) \right) + \frac{\eps}{2}\left| s(x) - g(x) \right| \right]p(x)~dx.
    \end{align*}

    We will minimize this quantity with respect to $s(x)$, subject to the constraint that $\int_\mathcal{X}s(x)p(x)~dx = 1$. We define the Lagrangian $\mathcal{L}(R,\lambda) = f(R) + \lambda (1-\int_{\mathcal{X}} s(x)p(x)~dx)$ and compute:
    \begin{align*}
        \frac{\partial \mathcal{L}}{\partial s}
        &= \int_\mathcal{X} \left[\log s(x) + 1 - \lambda + \frac{\eps}{2} \sign[s(x)-g(x)] \right]p(x)~dx,
    \end{align*}
    which yields the first-order optimality condition:
    \begin{equation*}
        \log s(x) + 1 - \lambda + \frac{\eps}{2} \sign[s(x)-g(x)] = 0
    \end{equation*}
    for all $x$. Rearranging, we conclude that:
    \begin{equation*}
        s(x) = \begin{cases}
            c_1 \coloneq e^{-\eps/2 + \lambda^* - 1} & g(x) < c_1, \\
            c_2 \coloneq e^{\eps/2+\lambda^*-1} & g(x) > c_2, \\
            g(x) & \text{otherwise},
        \end{cases}
    \end{equation*}
    where $\lambda^*$ is the unique value of the Lagrange multiplier such that $\int_{\mathcal{X}} s(x)p(x)~dx = 1$. From here, we define the sets:
    \begin{equation*}
        A = \lbrace x \mid q(x)/p(x) < c_1\rbrace,\quad B = \lbrace x \mid q(x)/p(x) > c_2\rbrace, \quad M = \mathcal{X} \setminus(A \cup B),
    \end{equation*}
    which yields the following expression for the density of the $\widetilde Q$:
    \begin{equation*}
        \widetilde q(x) = \begin{cases}
            c_1\cdot p(x) & x \in A, \\
            q(x) & x \in M, \\
            c_2 \cdot p(x) & x \in B.
        \end{cases}
    \end{equation*}
    We are now ready to define our e-variable, $E^*(x) \coloneq \frac{\widetilde q(x)}{p(x)}$. Immediately, we observe that $c_2/c_1 = e^\varepsilon$ and therefore that $E^*$ has bounded log-sensitivity of $\eps$. Moreover, we can compute that:
    \begin{align*}
        \E^\Alt[\log E^*] - \KL{\widetilde Q}{\Null}
        &= \int_A (q(x) - \widetilde q(x)) \log(c_1)~dx + \int_B (q(x) - \widetilde q(x)) \log(c_2)~dx \\
        &= \left( \widetilde Q(A) - \Alt(A)  \right)[\eps/2 - \lambda^* + 1] + (\Alt(B) - \widetilde Q(B))[\eps/2 + \lambda^* - 1] \\
        &= \eps \TV{\widetilde Q}{\Alt},
    \end{align*}

    where we used the fact that $\widetilde Q(A) - \Alt(A) = \Alt(B) - \widetilde Q(B) = \TV{\widetilde Q}{\Alt}$. It follows that:
    \begin{equation*}
        \E^\Alt[\log E^*] = \KL{\widetilde Q}{\Null} + \eps \TV{\widetilde Q}{\Alt},
    \end{equation*}

    as desired.
\end{proof}

\vspace{1cm}

We present the proof of \cref{thm:basic-batch}, which establishes that the private test statistic from \cref{eq:batch_test_statistic} matches the optimal rate up to lower-order terms.

\begin{proof}[Proof of \cref{thm:basic-batch}]
The non-private test statistic from \cref{eq:batch_test_statistic} has log-sensitivity:
    \begin{equation*}
        R_\lambda = \log \left( \frac{1-\lambda+\lambda c_2}{1-\lambda+\lambda c_1} \right).
    \end{equation*}
    In order to ensure that we end up with a valid e-variable, we need $\mathbb{E}[e^{Z_b}] = \frac{1}{1-b^2} < \infty$, which implies that $b < 1$. Meanwhile, to satisfy $\varepsilon$-DP, we require that $b \geq \frac{R_\lambda}{\varepsilon}$. Writing $\mu = \E^\Alt[\log E]$ and using the concavity of $\log(\cdot)$, we have that:
    \begin{equation*}
        \mathbb{E}^\Alt[\tilde \Lambda_S(\lambda)] = n \cdot\mathbb{E}^\Alt[\log(1-\lambda+\lambda E)] + \mathbb{E}[Z_b] - \log \mathbb{E}[\exp(Z_b)] \geq \lambda n \mu + \log (1-b^2).
    \end{equation*}

    We will choose $b = \frac{R_\lambda}{\varepsilon}$ exactly and reparameterize with $\lambda = 1 - s$. Taking the Taylor expansion of $R_s$ around $s=0$, we can approximate:
    \begin{align*}
        R_s &= \log\left( \frac{(1-s)c_2 + s}{(1-s)c_1 + s} \right) 
        = \log((1-s)c_2 + s) - \log((1-s)c_1 + s) \\
        \frac{\partial R_s}{\partial s} &=
        \frac{1-c_2}{(1-s)c_2 + s} - \frac{1-c_1}{(1-s)c_1 + s} \\
        &= \frac{(1-c_2)((1-s)c_1 + s) - (1-c_1)((1-s)c_2 + s)}{((1-s)c_2 + s)((1-s)c_1 + s)} \\
        &= \frac{s(1 - c_2 - 1 + c_1) + (1-s)(c_1 - c_1c_2 - c_2 + c_1c_2)}{((1-s)c_2 + s)((1-s)c_1 + s)} \\
        &= \frac{c_1-c_2}{((1-s)c_2 + s)((1-s)c_1 + s)} \\
        R_s &= \eps - \frac{c_2-c_1}{c_1c_2}s + O(s^2).
    \end{align*}
    Let $\beta = \frac{c_2-c_1}{c_1c_2} > 0$. Then we have that $b = 1 - s\beta/\eps + O(s^2/\eps)$ and thus that $\log(1-b^2) \approx \log(2s\beta/\eps) = \log(s) + O(1)$ for sufficiently small $s$. Plugging this approximation into our objective function, we get:
    \begin{equation*}
        (1-s)n\mu + \log s + O(1),
    \end{equation*}
    and applying the first order condition with respect to $s$, we arrive at an asymptotically optimal value of $s = \frac{1}{n\mu}$, which does indeed converge to 0. With this choice, we obtain that our objective value satisfies
    \begin{equation*}
        \mathbb{E}^\Alt[\tilde \Lambda_S(\lambda)] \geq \left(1-\frac{1}{n\mu}\right) n\mu + \log(1/(n\mu)) - O(1)
        = n\mu - \log(n\mu) - O(1),
    \end{equation*}
    as desired.
\end{proof}

\begin{remark}
    In practice, the optimal value of $\lambda$ (and $s$) should be computed directly through bisection. The particular value $s = \frac{1}{n\mu}$ is useful as a proof technique to capture the asymptotic behavior of the estimator, but is typically suboptimal for finite sample sizes.
\end{remark}

\vspace{1cm}

\subsection{Deferred proofs from \cref{sec:sequential}}
\label{sec:deferred_sequential}

The following helpful lemma will be used in the proof of \cref{thm:agr-ub}, and can be seen as a sequential analogue of \cref{prop:suboptimality}.

\begin{lemma}[Reduction to DP martingales]\label{lem:martingale-reduction}
        Let $(E_t)$ be an $\eps$-DP e-process adapted to the output filtration $\mathcal{F}_t = \sigma(E_1, \ldots, E_t)$. Then there exists an $\varepsilon$-DP $\Null$-martingale $(M_t)$ such that for any $\mathcal{F}$ stopping time $N$ which is $\Alt$-almost surely finite, we have:
        \begin{equation*}
            \E^\Alt[\log E_N] \leq \E^\Alt[\log M_N].
        \end{equation*}
        \end{lemma}

    \begin{proof}
        The construction in fact applies to any arbitrary $\varepsilon$-DP sequence $(O_t) = (\M(x_{1:t}))$ with corresponding output filtration $\mathcal{F}_t = \sigma(O_{1:t})$.
       Let $N$ be any $\Alt$-almost surely finite $\mathcal{F}$ stopping time, and let $A \in \mathcal{F}_N$ be any event such that $\Null(A)=0$. We have by group privacy that $\Alt(A) \leq e^{N\varepsilon} \Null(A) = 0$, and therefore that $\M(\Alt^N) \ll \M(\Null^N)$ on $\mathcal{F}_N$. The following likelihood ratio process therefore exists for all $t > 0$:
        \begin{equation*}
            M_t = \frac{d\M(\Alt^t)}{d\M(\Null^t)}(O_1, \ldots, O_t),
        \end{equation*}

        and by Proposition 3.22 of \citet{ramdasHypothesisTestingEvalues}, it is log-optimal in the sense that $\mathbb{E}^\Alt[\log M_N] \geq \mathbb{E}^\Alt[\log M'_N]$ for any other e-process $M'$ which is adapted to $\mathcal{F}$. Observing that it is also itself $\varepsilon$-DP by post-processing completes the proof. 
    \end{proof}

    \vspace{1cm}

\cref{thm:agr-ub} provides an upper bound on the expected value of a $\eps$-DP e-process $(E_t)$ at any random stopping time $N < N_{max}$.

\begin{proof}[Proof of \cref{thm:agr-ub}]

    We begin by establishing some notation. By \cref{lem:martingale-reduction}, we can assume without loss of generality that $E_t$ is a likelihood ratio process generated by the outputs of our DP mechanism. We will denote the corresponding log-likelihood ratio process by $L_t = \log \frac{d\M(\Alt^t)}{d\M(\Null^t)}(O_{1:t})$, which is adapted to the natural output filtration $\mathcal{F}_t = \sigma(O_{1:t})$.
    
    From here, we recall the total-variation coupling $\gamma$ between $\Alt$ and $\widetilde Q$ presented in \cref{sec:batch_optimal}: for $X, \tilde X \sim \gamma$ we have $X \sim \Alt$, $\tilde X \sim \widetilde Q$, and $w \coloneq \Pr[X_t \neq \tilde X_t] = TV(\widetilde Q, \Alt)$. Let $X_{1:n}, \tilde X_{1:n} \sim \gamma^n$ under the alternate, and let  $X_i = \tilde X_i \sim \Null$ under the null. In addition, we define a sequence of binary random variables $B_{1:n}$ such that $B_t = \mathbb{I}[X_i \neq \tilde X_i]$. Under the null, we define $B_t$ so that $\Null(B_t=1 \mid \tilde X_t) = \Alt(B_t = 1 \mid \tilde X_t)$.

    Finally, we introduce the shorthand $D(Y) = \KL{\Alt_Y}{\Null_Y}$, where $\Alt_Y$ (resp. $\Null_Y$) represents the distribution of $Y$ under $\Alt$ (resp. $\Null$). This  notation will prevent the following calculations from becoming unwieldy.
    
    With these preliminaries concluded, our basic proof strategy can be seen as an extension \cref{lem:decomp}: we will use the data-processing inequality to upper-bound $\E^\Alt[L_N] = D(O_{1:N})$ by the divergence $D(O_{1:N}, \tilde x_{1:N}, b_{1:N})$, which will turn out to be easier to analyze. Intuitively, the latter quantity represents the achievable e-power of a hypothetical mechanism which is given additional side information about the input, and this will yield a valid upper-bound on the achievable e-power of real-world mechanisms that lack that information. The first step in this process is to define an appropriate potential function $\Phi'$ which will allow us to largely ignore the internal workings of our particular mechanism $\M$.

    \paragraph{Constructing the potential function.}
    Formally, we define the augmented log-likelihood ratio process $L'_t = \log \frac{d\M(\Alt^t)}{d\M(\Null^t)}(O_{1:t}, B_{1:t})$, which is adapted to the finer filtration $\mathcal{G}_t = \sigma(O_{1:t}, B_{1:t})$. The conditional expected increments of $L'_t$ satisfy:
    \begin{equation*}
        \E^\Alt[L'_t - L'_{t-1} \mid \mathcal{G}_{t-1}] = D(O_t, B_t \mid \mathcal{G}_{t-1}).
    \end{equation*}
    We will also define the potential function $\Phi'_t = D(\tilde X_{1:t} \mid \mathcal{G}_{t})$. 
    
    Next, we apply the chain rule for KL divergence (\cref{lem:kl_chain_rule}) to the joint distribution of $O_t \otimes B_t, \tilde X_{1:t} $ conditioned on $\mathcal{G}_{t-1}$ to get:
    \begin{align*}
        D(O_t, \tilde X_{1:t}, B_t \mid \mathcal{G}_{t-1}) 
        &= D(O_t, B_t \mid \mathcal{G}_{t-1}) + \E^\Alt[D(\tilde X_{1:t}\mid \mathcal{G}_{t})] \\
        &= D(O_t, B_t \mid \mathcal{G}_{t-1}) + \E^\Alt[\Phi'_t].
    \end{align*}
    Separately, we apply the chain rule to the joint distribution of $O_t \otimes \tilde X_t \otimes B_t, \tilde X_{1:t-1}$ to compute:
    \begin{align*}
        D(O_t, \tilde X_{1:t}, B_t \mid \mathcal{G}_{t-1}) 
        &= D(\tilde X_{1:t-1} \mid \mathcal{G}_{t-1}) + \E^\Alt[D(O_t, \tilde X_t, B_t \mid \mathcal{G}_{t-1}, \tilde x_{1:t-1})] \\
        &= \Phi'_{t-1} + \E^\Alt[D(O_t, \tilde X_t, B_t \mid \mathcal{G}_{t-1}, \tilde x_{1:t-1})].
    \end{align*} 

     Because the left-hand sides of the two equations above are equal, we can equate the right-hand sides:
    \begin{equation*}
        D(O_t, B_t \mid \mathcal{G}_{t-1}) + \E^\Alt[\Phi'_t] = \Phi'_{t-1}  + \E^\Alt[D(O_t, \tilde X_t, B_t \mid \mathcal{G}_{t-1},\tilde x_{1:t-1})],
    \end{equation*}

    and by our formula for the expected conditional increments of $L'_t$, it follows that:
    \begin{equation*}
        \E^\Alt[L'_t + \Phi'_t \mid \mathcal{G}_{t-1}] = L'_{t-1} + \Phi'_{t-1} + \E^\Alt[D(O_t, \tilde X_t, B_t \mid \mathcal{G}_{t-1}, \tilde x_{1:t-1})].
    \end{equation*}

    Because KL divergence is non-negative, we conclude that the process $\tilde L'_t = L'_t + \Phi'_t$ is a submartingale under the alternate for the finer filtration $\mathcal{G}_t$. Also by non-negativity, we have $\tilde L'_t \geq L'_t$. 
    
    \paragraph{A fixed-time upper bound.} We now turn our attention to the increments. Once again applying the chain rule for KL divergence, we have:
    \begin{equation*}
        D(O_t, \tilde X_t, B_t \mid \mathcal{G}_{t-1}, \tilde x_{1:t-1}) 
        = D(\tilde X_t, B_t \mid \mathcal{G}_{t-1}, \tilde x_{1:t-1}) + \E^\Alt[D(O_t \mid \mathcal{G}_{t-1}, b_t, \tilde x_{1:t})].
    \end{equation*}

    We will bound each term separately. Since the $\tilde X_t,B_t$ pairs are i.i.d. and $B_t$ has the same conditional distribution given $\tilde x_t$ under the null and the alternate, we have that $D(\tilde X_t, B_t \mid \mathcal{G}_{t-1}, \tilde x_{1:t-1}) = D(\tilde X_t, B_t) = \KL{\widetilde Q}{\Null}$. 
    Simultaneously, we have that:
    \begin{align*}
    \sum_{i=1}^t  \E^\Alt[D(O_i \mid \mathcal{G}_{i-1}, b_i, \tilde x_{1:i})]
    &=
    \sum_{i=1}^t  \E^\Alt[D(O_i \mid \mathcal{F}_{i-1}, b_{1:t}, \tilde x_{1:t})]  \\
    &=
    D(O_{1:t} \mid b_{1:t}, \tilde x_{1:t}),
    \end{align*}
    where the first equality follows because $O_i$ is independent of $b_{i+1: t}, \tilde x_{i+1: t}$ and the second equality follows inductively from the chain rule for KL divergence (\cref{lem:kl_chain_rule}). Next, we apply joint convexity along with group privacy (see \cref{eq:jc+gp}) to conclude that:
    \begin{align*}
    D(O_{1:t} \mid b_{1:t}, \tilde x_{1:t})
    &\leq \sum_{i=1}^t b_i \varepsilon.
    \end{align*}

    
    Importantly, this is an almost sure upper bound on the sum of the expectations, and not merely a bound on the expectation of the sum. It follows that under $\mathcal{G}$, we can decompose $\tilde L'_t = W_t + A_t$ for some martingale $W_t = \sum_{i=1}^t w_i$ and some process $A_t = \sum_{i=1}^t a_i$ satisfying $A_t \leq \sum_{i=1}^t [\KL{\widetilde Q}{\Null} + b_i\varepsilon)]$ almost surely.

    \paragraph{Extension to bounded stopping times. }
    Now, let $N$ be any $\mathcal{F}$ stopping time which is upper-bounded by $N_{max} < \infty$. Then it is also a stopping time for the finer filtration $\mathcal{G}$. We can therefore write:
    \begin{align*}
        \E^\Alt[\tilde L'_N] &= \E^\Alt\left[ \sum_{t=1}^{N_{max}} (w_t + a_t)\lbrace N \geq t\rbrace  \right] \\
        &\leq \E^\Alt\left[ \sum_{t=1}^{N_{max}} (w_t + [\KL{\widetilde Q}{\Null} + b_t\varepsilon)])\lbrace N \geq t\rbrace  \right] \\
        &= \E^\Alt\left[ \sum_{t=1}^{N_{max}} [\KL{\widetilde Q}{\Null} + b_t\varepsilon)]\lbrace N \geq t\rbrace  \right],
    \end{align*}

    where the last line follows from linearity of expectation along with the optional stopping theorem (\cref{lem:optional_stopping}) applied to the martingale $W_t$, which is valid because $N \leq N_{max}$ almost surely. Now, we can take expectation over $b$. Because $\lbrace N \geq t\rbrace \in \mathcal{F}_{t-1} \perp b_t$, Wald's equation (\cref{lem:wald}) gives us that:
    \begin{equation*}
        \E^\Alt\left[ \sum_{t=1}^{N_{max}} \left(\KL{\widetilde Q}{\Null} + b_t\varepsilon\right)\lbrace N \geq t\rbrace  \right] = \E[N]\left(\KL{\widetilde Q}{\Null} + w \varepsilon\right),
    \end{equation*}

    and by rearranging, this gives us that:
    \begin{equation*}
        \E[N] \geq \frac{\E^\Alt[\tilde L'_N]}{\KL{\widetilde Q}{\Null} + w \varepsilon} \geq \frac{\E^\Alt[L'_N]}{\KL{\widetilde Q}{\Null} + w \varepsilon}.
    \end{equation*}
    
    This is nearly what we want, except that it's stated in terms of the sequence $L'_t$. To conclude, we observe that by the data processing inequality, we have that $\E^\Alt[L_N] = D(O_{1:N}) \leq D(\tilde x_{1:N}, b_{1:N}, O_{1:N}) = E^\Alt[\tilde L'_N]$, giving us finally that:
    \begin{equation*}
        \E[N] \geq \frac{\E^\Alt[L_N]}{\KL{\widetilde Q}{\Null} + w \varepsilon}.
    \end{equation*}

    Using the fact (proved in \cref{sec:batch}) that $\opt = \KL{\widetilde Q}{\Null} + w\eps$ then yields the theorem statement. This concludes the proof of \cref{thm:agr-ub}.
\end{proof}

\vspace{1cm}

\cref{prop:sequential-tests} applies \cref{thm:agr-ub} to lower-bound the expected stopping time of any $\eps$-DP sequential test, even those that are not obviously based on e-processes.

\begin{proof}[Proof of \cref{prop:sequential-tests}]
    Let $\M: \mathcal{X}^* \to \lbrace 0, 1 \rbrace^*$ be any $\eps$-DP sequential test with power $1-\beta$ and level $\alpha < 1-\beta$. We will write the output process as $\phi_1, \phi_2, \ldots$, interpreting $\phi_t = 1$ to mean that $\M$ rejected the null hypothesis at or before time $t$. Let $N$ be the stopping time of $\M$, and assume that it is bounded above by an arbitrary constant $N_{max} < \infty$ (this is satisfied by all real-world tests). For $n > N$, we will extent the output of $\M$ by setting $\phi_n = \phi_N$. With this notation, the condition that $\M$ has power $1-\beta$ and level $\alpha$ means that $\Alt(\phi_N = 1) = 1-\beta$ and $\Null(\phi_N = 1) = \Null(\exists t\,.\,\phi_t = 1) = \alpha$. 

    Now, define the likelihood ratio process $M_t = \frac{d \M(\Alt^t)}{d\M(\Null^t)}(\phi_{1:t})$. In the proof of \cref{lem:martingale-reduction}, we show that $M_t$ is an $\eps$-DP e-process. So, by \cref{thm:agr-ub} along with the definition of KL divergence, we have the following upper bound:
    \begin{equation*}
        \KL{\Alt_{\phi_{1:N}}}{\Null_{\phi_{1:N}}} = \E^\Alt[\log M_N] \leq \E^\Alt[N] \opt.
    \end{equation*}
    On the other hand, by the data-processing inequality along with our assumption that $\M$ has power $1-\beta$ and level $\alpha$, we have the lower bound:
    \begin{align*}
        \KL{\Alt_{\phi_{1:N}}}{\Null_{\phi_{1:N}}} 
        &\geq \KL{\Alt_{\phi_{N}}}{\Null_{\phi_{N}}} \\
        &= \KL{\text{Bern}(1-\beta)}{\text{Bern}(\alpha)} \\
        &= (1-\beta)\log((1-\beta)/\alpha) + \beta\log(\beta/(1-\alpha)).
    \end{align*}

    Rearranging, we conclude that:
    \begin{equation*}
        \E[N] \geq \frac{(1-\beta)\log((1-\beta)/\alpha) + \beta\log(\beta/(1-\alpha))}{\opt},
    \end{equation*}

    completing the proof. 
    \end{proof}

    We remark that a `converse' to the construction in the preceding proof holds as well: given any $\eps$-DP e-process $E$, we can post-process it into a $\eps$-DP level-$\alpha$ sequential test by defining $\phi_t = \mathbb{I}[\exists i \leq t\,.\,E_i \geq 1/\alpha]$. In other words, any hardness result for $\eps$-DP sequential tests can be immediately turned into a hardness result for $\eps$-DP e-processes, and visa versa.

\vspace{1cm}

\begin{proof}[Proof of \cref{lem:ind-martingale}]
    We can decompose the ratio $E_t/E_{t-1}$  into:
    \begin{equation*}
        \exp(\lambda(\tilde \Lambda_t - \tilde\Lambda_{t-1}) - K_t(\lambda)) = \exp(\lambda(\Lambda_t - \Lambda_{t-1})) \exp(\lambda(\xi_t - \xi_{t-1}) - K_t(\lambda)).
    \end{equation*}
    Because $M_t$ is a test (super)martingale, we know that $\E^\Null[\exp(\Lambda_t - \Lambda_{t-1}) \mid \Lambda_{1:t-1}] = \E^\Null[M_t/M_{t-1} \mid M_{1:t-1}] \leq 1$. Then, for $\lambda \in [0,1]$, the map $x \mapsto x^\lambda$ is concave, and so Jensen's inequality gives us that $\E^\Null[\exp(\lambda(\Lambda_t - \Lambda_{t-1})) \mid \Lambda_{1:t-1}] \leq \E^\Null[\exp(\Lambda_t - \Lambda_{t-1}) \mid \Lambda_{1:t-1}]^\lambda \leq 1$. Applying the definition of $K_t(\lambda)$, it follows that $\exp(\lambda(\xi_t - \xi_{t-1}) - K_t(\lambda)) \leq 1$ as well. Finally, we use independence to conclude that $E_t$ is a test (super)martingale for $\M(\Null^\mathbb{N})$ and thus an $\eps$-DP e-process.
\end{proof}

\vspace{1cm}

\begin{proof}[Proof of \cref{thm:sequential-lb}]
The fact that \cref{alg:batch2proc} satisfies $\eps$-DP follows from standard analysis of the Laplace mechanism. To see that it is an e-process, let $\mathcal{G}_j = \mathcal{F}_{t_j}$ denote the `coarsened' filtration seen by the batched process, so that $\Lambda_j - \Lambda_{j-1} = \sum_{i=t_{j-1}+1}^{t_j} \log E(x_i)$. By \cref{lem:ind-martingale}, $(\widetilde E_{t_j})$ is an $\eps$-DP test (super)martingale for $\mathcal{G}_j$. Expanding back out to $\mathcal{F}$ clearly preserves the martingale property since $\widetilde E_i = \widetilde E_{i-1}$ exactly for any $i \neq t_j$, and so we conclude that $(\widetilde E_i)$ is a test (super)martingale for the original output filtration $\mathcal{F}$ as well.

To analyze its expected stopping times, define $J(t) = \max \lbrace j : t_j \leq t \rbrace$ to be the number of batch updates up to time $t$. Let $C_\lambda = -\log(1-c^2\lambda^2)$. Our goal is to minimize the following competitive ratio between the original non-private process and our output process:
\begin{equation*}
    \sup_{N \geq t_1} \frac{\E^\Alt[\Lambda_N]}{\E^\Alt[\log \widetilde E_N]} =\sup_{N \geq t_1} \frac{\E^\Alt[\Lambda_N]}{\lambda \E^\Alt[\Lambda_{t_{J(N)}}] - J(N) C_\lambda}.
\end{equation*}

We remark that the minimum time $t_1$ is necessary because the competitive ratio would otherwise be $\infty$ at early time steps when we don't have enough data to release anything. With this constraint, the supremum will always occur when $N=t_{J(n)+1}-1$ for some $n>0$, in which case the expression simplifies to:
\begin{equation*}
    \sup_{n>0} \frac{n \mu}{t_{J(n)} \lambda\mu  - J(n)C_\lambda} = \sup_{j>0} \frac{(t_{j+1}-1)\mu}{t_j\lambda\mu - jC_\lambda}.
\end{equation*}

To minimize the supremum, we would like to choose a batch schedule so that the worst-case relative error is equal to some constant $\rho>1$ for all $j$. Asymptotically, $t_{j+1} -1 \sim t_{j+1}$, and so we get the following recurrence:
\begin{align*}
    t_{j+1} = \rho \left(\lambda  t_j - \frac{ jC_\lambda}{\mu} \right).
\end{align*}

We will solve this recurrence using the method of undetermined coefficients. The homogeneous equation $t_{j+1} = \rho\lambda t_j$ has solution $t_{j+1} = C(\rho\lambda)^{j+1}$, so we'll guess that the solution to the equation has the form:
\begin{equation*}
    t_{j+1} = C(\rho\lambda)^{j+1} + \alpha j + \gamma.
\end{equation*}

Writing $\beta = \rho C_\lambda/\mu$, this then implies that:
\begin{align*}
    \alpha j + \gamma &= \rho\lambda(\alpha(j-1)+\gamma) - \beta j \\
    &= (\rho\lambda\alpha-\beta)j + \rho\lambda\gamma - \rho\lambda\alpha,
\end{align*}
and thus, equating coefficients, that:
\begin{align*}
    \alpha &= \rho\lambda\alpha-\beta \\
    \alpha &= \frac{\beta}{\rho\lambda-1},
\end{align*}

and
\begin{align*}
    \gamma &= \rho\lambda(\gamma-\alpha) \\
    \gamma &= \frac{\rho\lambda\alpha}{\rho\lambda-1} = \frac{\rho\lambda\beta}{(\rho\lambda-1)^2}.
\end{align*}

Let $R = \rho\lambda$. Then our candidate solution is:
\begin{equation*}
    t^*_{j+1}=CR^{j+1} + \frac{\beta j}{R-1} + \frac{R\beta}{(R-1)^2}.
\end{equation*}
Plugging this in to our original recurrence to verify, we get:
\begin{align*}
    Rt^*_j - \beta j
    &= CR^{j+1} + \frac{R\beta(j-1)}{R-1} + \frac{R^2\beta}{(R-1)^2} - \beta j \\
    &= CR^{j+1} + \left(1 + \frac{1}{R-1} \right)\beta(j-1) + \left(1 + \frac{1}{R-1} \right)^2\beta - \beta j \\
    &= CR^{j+1} + \left(\frac{1}{R-1} \right)\beta j - \left(1+\frac{1}{R-1} \right)\beta  + \left(1 + \frac{1}{R-1} \right)^2\beta \\
    &= CR^{j+1} + \left(\frac{1}{R-1} \right)\beta j + \left(\frac{1}{R-1} + \frac{1}{(R-1)^2} \right)\beta \\
    &= CR^{j+1} + \left(\frac{1}{R-1} \right)\beta j + \frac{R\beta }{(R-1)^2} \\
    &= t^*_{j+1},
\end{align*}

confirming that our solution is valid. Expanding the variables we introduced and setting $C=1$ for simplicity, this gives us:
\begin{align*}
    t_1 &= \rho\lambda + \frac{\rho^2\lambda \log(1/(1-c^2\lambda^2))}{\mu (\rho\lambda-1)^2} \\
    t_j &= (\rho\lambda)^j + \frac{\rho \log(1/(1-c^2\lambda^2)) (j-1)}{\mu(\rho\lambda-1)} + \frac{\rho^2\lambda \log(1/(1-c^2\lambda^2))}{\mu(\rho\lambda-1)^2}.
\end{align*}

Returning finally to our original stopping time equation, we get that with this batching schedule and for any $N \geq t_1$,

\begin{equation*}
    \E^\Alt[\log \widetilde E_N] \geq \frac{\E^\Alt[\Lambda_N]}{\rho} = \frac{\E[N]\mu}{\rho}.
\end{equation*}

Since $\lambda \in (0, \min(1,1/c))$ and we require $\rho\lambda > 1$ for feasibility, this implies that we can achieve any competitive ratio $\rho \in (\max(1,c),\infty)$, but at the cost of a larger initial startup time. To understand this tradeoff, we examine the asymptotic behavior of $t_1$ as $\rho \to c$ for $c \geq 1$. Choosing the (potentially suboptimal) value $\lambda = 1/\sqrt{\rho c} \in (1/\rho, 1/c)$ and letting $\eta = \sqrt{\rho/c}$, we get that:
    \begin{equation*}
        t_1 = \eta - \frac{\rho}{\mu} \left( \frac{\eta \log(1-\frac{1}{\eta^2})}{(\eta-1)^2} \right)
        = \eta - \frac{\rho}{\mu} \left( \frac{\eta (\log(\eta-1) + \log(\eta+1) - 2\log(\eta))}{(\eta-1)^2} \right).
    \end{equation*}

Now, letting $\eta = 1+\delta$ and taking the limit as $\delta \to 0^+$, we get:
\begin{equation*}
    1 - \frac{\rho}{\mu} \left( \frac{\log(\delta) + \log(2)}{\delta^2} \right) = O\left( \frac{\rho \log(1/\delta)}{\mu\delta^2} \right) = \tilde O\left( \frac{c}{\mu(\eta-1)^2} \right),
\end{equation*}

where we used the fact that $\rho \to c$ as $\eta \to 1$. Finally, since this upper-bound holds for our fixed choice of $\lambda$, it also holds for the optimal choice.
\end{proof}

\end{document}